\documentclass[journal]{IEEEtran}

\usepackage{amsmath,amsfonts}
\usepackage{amsthm}

\usepackage{algpseudocode}
\usepackage{algorithm}

\usepackage{graphicx}
\usepackage{hyperref}
\usepackage{multirow}
\usepackage{soul}
\usepackage{subcaption}
\usepackage{tabularx}
\usepackage{textcomp}
\usepackage{xcolor}
\usepackage{xspace}
\usepackage{multirow}
\usepackage{booktabs}
\usepackage{graphicx}
\usepackage{threeparttable}  
\usepackage{amsmath}
\usepackage{amsfonts}
\usepackage{tcolorbox}
\usepackage{stmaryrd}
\usepackage{pifont}

\usepackage[normalem]{ulem}
\usepackage{enumitem}

\newtheorem{definition}{Definition}
\newtheorem{theorem}{Theorem}

\newtheorem{corollary}{Corollary}

\newcommand{\Ra}{\ensuremath{\stackrel{\$}{\leftarrow}{\xspace}}}

\definecolor{mygreen}{RGB}{112, 173, 71}
\definecolor{myred}{RGB}{192, 0, 0}
\definecolor{kellygreen}{rgb}{0.3, 0.73, 0.09}

\newcommand{\cut}[1]{}

\newcommand{\kem}{\ensuremath{\mathcal{E}}}
\usepackage{threeparttable}  
\usepackage{booktabs}
\usepackage{multirow}
\usepackage{bm}

\algdef{SE}[DOWHILE]{Do}{doWhile}{\algorithmicdo}[1]{\algorithmicwhile\ #1}%

 \usepackage{tikz}
 \usetikzlibrary{plotmarks}
 \usepackage{pgfplots}
 \usetikzlibrary{patterns}
 \pgfplotsset{compat=1.3}
 \usetikzlibrary{backgrounds}

\newcommand{\sk}{\ensuremath{\mathsf{sk}}}
\newcommand{\pk}{\ensuremath{\mathsf{pk}}}
 \newcommand{\sgnkeygen}{{\mathtt{KeyGen}}}
 \newcommand{\sgnsign}{{\mathtt{Sign}}}
 \newcommand{\presgnsign}{{\mathtt{PSgn}}}
 \newcommand{\precomputesgn}{{\mathtt{Precmp}}}
  
  \newcommand{\sgnverify}{{\mathtt{Verify}}}

     \newcommand{\kagen}{{\mathtt{KeyGen}}}
   \newcommand{\kemenc}{{\mathtt{Encaps}}}
 \newcommand{\kemdec}{{\mathtt{Decaps}}}
  \newcommand{\x}{{\mathtt{x}}}

     \newcommand{\encap}{{\mathtt{Encp}}}
          \newcommand{\decap}{{\mathtt{Decp}}}

\newcommand{\sys}{\ensuremath{\texttt{Beskar}}\xspace}
\newcommand{\micro}{MicroFedML\xspace}
\newcommand{\microfirst}{MicroFedML1\xspace}
\newcommand{\microsecond}{MicroFedML2\xspace}
\newcommand{\flamingo}{Flamingo\xspace}
\newcommand{\pqsa}{PQSA\xspace}
\newcommand{\eseafl}{e-SeaFL\xspace}

\newcommand{\node}{\ensuremath{\mathtt{A}}}

\newcommand{\agg}{\ensuremath{\mathtt{S}}}

\newcommand{\user}{\ensuremath{\mathtt{P}}}

\newcommand\mal[1]{\textcolor{red}{#1}}

\renewcommand{\vec}[1]{\ensuremath{\mathbf{#1}}}

\newcommand{\prf}{\ensuremath{\mathtt{PRF}}}

\newcommand{\Ulist}{\ensuremath{\mathcal{L}}}

\newcommand{\Sim}{\ensuremath{\mathtt{Sim}}}
\newcommand{\ppp}{\ensuremath{\mathcal{P}}}
\newcommand{\aaa}{\ensuremath{\mathcal{A}}}

\newcommand{\RO}{\ensuremath{\mathtt{RO}}}

\newcommand{\real}{\ensuremath{\mathtt{REAL}}}

 \newcommand{\indist}{{{indistinguishability}}}

 \usepackage{xstring}
 
 \newcommand{\PC}[1]{
 	\vspace{2px}
 	\noindent{\bf \IfEndWith{#1}{:}{#1}{#1:}}
 }

\usepackage{etoolbox}
\newcommand{\circled}[2][]{%
	\tikz[baseline=(char.base)]{%
		\node[shape = circle, draw, fill=red, color=red, inner sep = .2pt]
		(char) {\phantom{\ifblank{#1}{#2}{#1}}};%
		\node at (char.center) {\makebox[0pt][c]{\color{white}{#2}}};}}
\robustify{\circled}

\newcommand{\Kg}{\ensuremath{\mathsf{Kg}}}
\newcommand{\Sgn}{\ensuremath{\mathsf{Sgn}}}
\newcommand{\Vfy}{\ensuremath{\mathsf{Vfy}}}
\newcommand{\Hash}{\ensuremath{\mathsf{Hash}}}
\newcommand{\Rssh}{\ensuremath{\mathsf{Rssh}}}
\newcommand{\Rsshe}{\ensuremath{\mathsf{Rssh}_\mathsf{e}}}
\newcommand{\sshl}{\ensuremath{\mathsf{ssh}_\mathsf{l}}}
\newcommand{\PRG}{\ensuremath{\mathsf{PRG}}}
\newcommand{\PRF}{\ensuremath{\mathsf{PRF}}}
\newcommand{\AEnc}{\ensuremath{\mathsf{AEnc}}}
\newcommand{\ADec}{\ensuremath{\mathsf{ADec}}}
\newcommand{\SEnc}{\ensuremath{\mathsf{SEnc}}}
\newcommand{\SDec}{\ensuremath{\mathsf{SDec}}}
\newcommand{\PIL}{\ensuremath{\mathsf{PI}_\mathsf{L}}}
\newcommand{\Mu}{\ensuremath{\mathsf{mu}}}
\newcommand{\Ex}{\ensuremath{\mathsf{ex}}}
\newcommand{\Su}{\ensuremath{\mathsf{su}}}
\newcommand{\Ka}{\ensuremath{\mathsf{Ka}}}

\newcommand{\PSgn}{\ensuremath{\mathsf{PSgn}}}
\usepackage{pifont}
\newcommand{\cmark}{\ding{51}}%
\newcommand{\xmark}{\ding{55}}%

\definecolor{ForestGreen}{rgb}{0.0, 0.5, 0.0}
\newcommand{\yes}{{\color{ForestGreen}\cmark}}
\newcommand{\no}{{\color{red}\xmark}}

\ifCLASSINFOpdf
\else
\fi
\begin{document}
%
\title{Efficient Full-Stack Private Federated Deep Learning with Post-Quantum Security}

%
%
%

\author{Yiwei Zhang, Rouzbeh Behnia, Attila A. Yavuz, Reza Ebrahimi, Elisa Bertino
\thanks{Yiwei Zhang and Elisa Bertino are with Purdue University. E-mail: {yiweizhang, bertino}@purdue.edu.}
\thanks{Rouzbeh Behnia, Attila A. Yavuz, Reza Ebrahimi are with the University of South Florida. E-mail: {behnia, attilaayavuz, ebrahimim}@usf.edu.}
}

\maketitle

\begin{abstract}

Federated learning (FL) enables collaborative model training while preserving user data privacy by keeping data local. Despite these advantages, FL remains vulnerable to privacy attacks on user updates and model parameters during training and deployment. Secure aggregation protocols have been proposed to protect user updates by encrypting them, but these methods often incur high computational costs and are not resistant to quantum computers. Additionally, differential privacy (DP) has been used to mitigate privacy leakages, but existing methods focus on secure aggregation or DP, neglecting their potential synergies. To address these gaps, we introduce \sys, a novel framework that provides post-quantum secure aggregation, optimizes computational overhead for FL settings, and defines a comprehensive threat model that accounts for 
a wide spectrum of adversaries. We also integrate DP into different stages of FL training to enhance privacy protection in diverse scenarios. Our framework provides a detailed analysis of the trade-offs between security, performance, and model accuracy, representing the first thorough examination of secure aggregation protocols combined with various DP approaches for post-quantum secure FL. \sys aims to address the pressing privacy and security issues FL while ensuring quantum-safety and robust performance.

\end{abstract}


%
\IEEEpeerreviewmaketitle

\hfill  
 
\hfill

\section{Introduction}
Federated learning (FL) enables collaborative learning of a shared model between distributed parties while keeping the data local, mitigating data privacy and collection challenges common in traditional centralized learning.  
In large-scale FL, clients with limited computational resources, such as mobile devices, can contribute to training a global model with the assistance of a central server. 
In each iteration, the central server distributes an \textit{intermediate model} to all clients, who then train the model using their local data to compute \textit{local updates} (i.e., user gradients). 
The server aggregates the local updates from all users, refines the intermediate model, and sends it back to the clients. 
This iterative process continues until the model achieves a satisfactory level of performance, resulting in a \textit{final model} to be deployed in production.
FL contributes to data privacy by keeping the user data local.  
However, recent attacks have demonstrated that deploying a plain FL paradigm is insufficient to protect the privacy of the participating users' data~\cite{carlini2021extracting, shokri2017membership, carlini2019secret}. More specifically, these attacks can undermine the privacy of the training data by having access only to the user updates or the model at any stage (training and/or deployment). 
A well-known solution to protect the user updates during the training phase is secure aggregation~\cite{bonawitz2017practical}, where the server can compute the global model without knowledge of any individual user update. This is achieved by masking/encrypting the updates so that the masking factors cancel out during aggregation. 
Secure aggregation can be achieved with different techniques, such as secure distributed computation~\cite{rathee2022elsa, behnia2023efficient,bell2020secure} or Homomorphic Encryption (HE)~\cite{truex2019hybrid,elahi2014privex}. 
However, existing secure aggregation protocols often incur high communication and/or computation overhead.  

Additionally, to our knowledge, except for a few~(e.g., \cite{yang2022post}), the existing methods are primarily based on conventional cryptographic tools. 
However, such tools are not resistant to quantum computers, which are on the verge of becoming a reality. 
Therefore a pressing requirement is that FL protocols (and other distributed protocols for machine learning) must be post-quantum secure. Given the directives by the NSA and the White House~\cite{chen2016report,nist2024pqc,wh2024report}, NIST has suggested a series of post-quantum (PQ) secure cryptographic schemes.
%
%
While one might consider the direct adoption of the recent NIST PQ cryptographic standards~\cite{nistpqalgo}, these schemes and their subsequent extensions (e.g., \cite{conf/ccs/BehniaOYR18, conf/ACSAC/BehniaY21}), despite their elegant designs, might not be suitable for highly distributed settings (e.g., FL) with resource- and bandwidth-constrained devices (e.g., mobile phones).
Finally, existing methods do not protect the intermediate model distributed by the server in each iteration, making it vulnerable to privacy attacks by adversaries disguised as clients.




Differential privacy (DP)~\cite{dwork2006differential}, as a popular statistical tool, can mitigate these privacy leakages effectively. 
This is achieved by injecting a controlled noise to the model to distort the effects of individual data points on model parameters. 
Abadi et al.~\cite{AbadiCGMMT016} introduced the concept of DP in deep learning by proposing DP-SGD, a privacy-preserving version of the well-known SGD algorithm.
In the FL setting, DP  can be applied independently or in conjunction with secure aggregation at various stages of the training process (e.g., ~\cite{chen2021voltpillager, van2018foreshadow, chamikara2022local, truex2020ldp, kairouz2021advances, balle2020privacy, ramaswamy2020training}), taking into account different performance impacts and adversarial models. Accordingly, implementing these privacy-preserving techniques often depends on thoroughly understanding the adversarial and threat models involved. Existing approaches have primarily been designed focusing on individual privacy-preserving methods (i.e., secure aggregation or DP); therefore, there remains a significant gap regarding a comprehensive threat model tailored for privacy-preserving FL. 
We stress that establishing a comprehensive threat model and evaluating the effectiveness of recommended privacy-preserving methods for each threat scenario would be critical for organizations striving to address diverse privacy requirements and comply with regulations such as HIPAA.

\begin{table}[ht]
    \centering
    \small
    \caption{High-level Comparison with State-of-the-Art}
    \label{tab:summary}
    \begin{threeparttable}
    \begin{tabular}{|c|c|c|c|c|}\hline 
        \textbf{Protocol} & \textbf{Rd.}  &  \begin{tabular}{@{}c@{}c@{}}
		 \textbf{Dropout} \\ \textbf{Resilience} \end{tabular}  & 
   \textbf{Model Privacy}
   & \textbf{PQ} \\\hline 
         \cite{bell2020secure} & 6 & Low&  \no & \no \\\hline 
       \cite{flamingo}  & 3 & Moderate& \no &  \no  \\ \hline  
        \cite{microfedml} & 3 & Low& \no &  \no  \\\hline 
       \cite{yang2022post} & 3 & Low& \no &  \yes \\ \hline
       \cite{behnia2023efficient} & 1 & High &  \no & \no \\ \hline
       Ours  & 1 & High & \yes & \yes \\\hline 
    \end{tabular}
\end{threeparttable}

\end{table}

\subsection{Our Contribution} \label{subsec:OurContrib}


In response to the above requirements, we propose \sys. To our knowledge, \sys~is the first to introduce a comprehensive threat model for FL settings by considering adversaries with different capabilities. 
As shown in Table~\ref{tab:summary}, \sys is the only solution offering high dropout resilience with only one communication round, while simultaneously ensuring post-quantum (PQ) security
and 
privacy of user data
during and after training. 
\sys~{\em fills the critical gaps in existing protocols by its balance of privacy, minimal communication requirements, and robustness against quantum attacks, making it a holistic solution for FL in the post-quantum era}.
We detail the contribution of our work in what follows.


\vspace{5pt}
\noindent
$\bullet$~\textbf{Efficient Post-Quantum Secure Aggregation.} We propose a new secure aggregation framework with post-quantum security by leveraging NIST post-quantum standards. Despite their elegant design, most of the suggested standards 
(e.g.,~\cite{ducas2018crystals,bos2018crystals}) are not optimized for resource-constrained environments, such as mobile devices, where minimizing computational overhead is critical.
While prior works 
(e.g.,~\cite{yavuz2013eta}) have demonstrated the effectiveness of pre-computation techniques for minimizing the overhead for mobile devices, applying these methods to their post-quantum counterparts is not straightforward due to their inherent design requirements (e.g., rejection sampling). Other methods 
(e.g.,~\cite{Tachyon}) rely on precomputed tables that might not be suitable for mobile devices due to their significant storage overhead. 
To address these challenges, we develop two lightweight yet efficient pre-computation strategies that explicitly account for the rejection sampling of post-quantum methods and eliminate the need for large lookup tables. Our methods are specifically designed for low-end devices (e.g., mobile devices) with limited computational and storage resources, ensuring practical deployment in resource-constrained environments.
Both optimizations leverage the characteristic that FL operates over a relatively small number of iterations, typically in the order of hundreds. 
Our first optimization algorithm improves the signature generation of Dilithium~\cite{dilithium}  by pre-computing message-independent tokens, thereby minimizing the signature generation overhead. 
This results in a 30\% faster signing process compared to the standard Dilithium algorithm.
Generating the masking terms to hide each element of user gradient is one of the dominant costs in secure aggregation protocols designed for larger deep learning models. Our second optimization algorithm significantly improves this process by achieving favorable computation and storage trade-offs and pre-computing a masking table for each client. 
Our experiments demonstrate that the two optimizations yield efficiency improvements of 134x, 1.1x, and 1233x in the aggregation phase with 1,000 clients.


\noindent $\bullet$~\textbf{Comprehensive Threat Models for FL.}  Existing approaches~\cite{flamingo,microfedml,yang2022post,sun2020ldp,geyer2017differentially} typically assume a single type of adversary, whereas FL applications face diverse threats and have to meet varying security requirements depending on their specific use cases. 
We thus propose the first comprehensive security framework for FL settings by considering three types of adversaries with varying capabilities and access levels. 
We define three distinct threat models corresponding to these adversaries, each posing a unique threat to the privacy of local updates, intermediate models, and final models, respectively.
This comprehensive approach enables developers and organizations to identify the threat model most pertinent to their privacy needs or mandated by regulatory frameworks such as HIPAA. Thus, a customized security solution is provided that aligns with specific compliance obligations and protection needs and guarantees robust and targeted privacy safeguards. \vspace{1mm}

\noindent $\bullet$~\textbf{Various Compositions of Secure Aggregation with DP.} 
Existing approaches~\cite{bell2020secure,yang2022post,yang2023privatefl} primarily focus on enhancing secure aggregation using various strategies (e.g., MPC, HE) or applying DP methods separately. However, these approaches do not offer comprehensive privacy protections for FL, particularly against the threat models discussed above.
To address this gap, we integrate different DP applications within FL to provide tailored privacy protection against various threat models. Our approach involves applying DP at different stages of FL training, effectively defending against the adversaries defined in our threat models.
Additionally, we conduct a thorough performance analysis of our approach, emphasizing the trade-offs between security, performance, and accuracy. To our knowledge, this is the first work to analyze the performance of secure aggregation protocols using different DP approaches, offering valuable insights into their effectiveness across various scenarios.

\section{Preliminary}


Vectors are denoted by bold letters (e.g., $\Vec{a}$). We define a pseudorandom function $\prf: \{0,1\}^*\rightarrow \Vec{a}  $ where $\Vec{a}$ shares the same dimension as the global model. 
The number of elements in a list  $\mathcal{L}$ is represented as  $|\mathcal{L}|$. Given two vectors $\Vec{a}$ and $\Vec{b}$ with the same dimension, $\Vec{a}+\Vec{b}$ denotes the element-wise addition. $[j]$ denotes the set $\{1,\dots,j\}$. In the following, we define the cryptographic building blocks used in our protocol.

\begin{definition}[Key-encapsulation \cite{bos2018crystals}]\label{def:kem}
A key-encapsulation mechanism  $\kem=\{\sgnkeygen, \kemenc, \kemdec\}$  with key space $\mathcal{K} $ is defined as follows. 
\begin{itemize}
    \item $(\pk_\kem,\sk_\kem)\gets \sgnkeygen(1^\kappa)$: On the input of the security parameter $\kappa$, it returns a pair consisting of a public and private key $(\pk_\kem,\sk_\kem)$. 
    \item $(c_x, x)  \gets \kemenc(\pk_\kem)$: On the input of the public key, it returns a key $x\in \mathcal{K}$ and its ciphertext $c_x$. 
    \item $\{\bot, x\}\gets \kemdec(\sk_\kem,c_x)$: On the input of the secret key and ciphertext, it either outputs the key $x\in \mathcal{K}$ or $\bot$, indicating rejection. 
\end{itemize}
  
\end{definition}
 
A key encapsulation algorithm is $(1-\beta)$-correct if $\Pr(c_x \gets \kem.\kemdec(\sk_\kem,c_x)$ :$(c_x,x) \gets \kem.\kemenc(\pk_\kem))\geq 1-\beta$ where probability is taken over $\kem.\sgnkeygen(\cdot)$ and  $\kem.\kemenc(\cdot)$. 
The standard security notion for a key-encapsulation algorithm is indistinguishably under a chosen-ciphertext attack (IND-CCA) where the adversary has access to a $\kem.\kemdec(\cdot)$ oracle.

\begin{definition}[Digital Signatures]\label{def:sig}
A digital signature scheme  for a messages space $\mathcal{M}_{\Pi}$ is defined by $\Pi$:$(\sgnkeygen$, $\sgnsign$, $\sgnverify)$:

\begin{itemize}

\item $(\pk_\Pi,\sk_\Pi)\gets\sgnkeygen(1^\kappa)$: On the input of security parameter $\kappa$, it returns a pair consisting of a public and private key $(\pk_\Pi,\sk_\Pi)$.

\item $\sigma \gets \sgnsign(\sk_\Pi,m)$:  On the input of a secret private key $\sk_\Pi$ and a message $m\in\mathcal{M}_{\Pi}$, it outputs a signature $\sigma$.
\item $\{0,1\}\gets \sgnverify(\pk_\Pi,m,\sigma)$:  On the input of a public key $\pk_\Pi$, a message $m \in\mathcal{M}_{\Pi}$, and an alleged signature $\sigma$, it outputs $1$ if $\sigma$ is a valid signature under $\pk_\Pi$ for message $m$. Otherwise, it outputs $0$.
\end{itemize}

\end{definition}
A digital signature scheme is correct if for any $m\in\mathcal{M}_{\Pi}$, $\sgnverify(\pk_\Pi,m,\sigma)=1$, where $(\pk_\Pi,\sk_\Pi)\gets\sgnkeygen(1^\kappa)$ and $\sigma\gets\sgnsign(\sk_\Pi, m)$.
The standard security notion for a digital signature scheme is existential unforgeability against adaptive chosen message attacks (EU-CMA), defined in Appendix \ref{sec:cryptonotions}.

Following Bell et al.~\cite{Bell2020}, we utilize the following $\alpha$-summation ideal functionality to prove the privacy of our protocol. The $\alpha$-summation ideal functionality necessity for a sufficient proportion of honest clients is to ensure that the aggregated model in the secure aggregation setting does not leak any information about each user update. 

\begin{definition}[$\alpha$-summation ideal functionality \cite{Bell2020}]\label{def:alphasummation}

  Given  $p,n,d$ as integers, we let $L\subseteq [n]$ and $\mathcal{X}_L:=\{\Vec{x}_i\}_{i\in {L}}$ where $\Vec{x}_i \in\mathbb{Z}_p^d$. Now, given a $0 \leq \alpha \leq 1$ and  $Q_{L}$ as the set of partitions of $L$  and a set of pairwise disjoint subsets  $ \{L_{1}, \dots, L_l\} \in Q_{L}$, the $\alpha$-summation ideal functionality $\mathcal{F}_{\vec{x},\alpha}(\cdot)$
computes $\mathcal{F}_{\vec{x},\alpha}(\{L_i\}_{i\in[1,\dots,l]}) \rightarrow  \{\Vec{s}_i\}_{i\in[1,\dots,l]}$
where 
\begin{equation*}
  \forall j\in[1,\dots,l] ,\Vec{s}_{j} =
    \begin{cases}
       \sum_{j\in Q_L}\vec{x}_j & \text{if}~Q_L|\geq \alpha|L| \\
      
      \bot & \text{else.}
    \end{cases}       
\end{equation*}

\end{definition}

\begin{definition}[The MSIS Problem~\cite{ducas2018crystals}]\label{def:msis}
With an algorithm $\mathcal{A}$, we associate the advantage function $Adv_{m, k, \gamma}^{MSIS}$ to solve the $MSIS_{m, k, \gamma}$ problem over the ring $R_q$ as 


\begin{align*}
  Adv_{m, k, \gamma}^{MSIS}(A) := 
  &\mathbb{P}(0 < \|y\|_\infty \leq \gamma \wedge (\mathcal{I} | \mathcal{A}) \cdot y = 0 \\
  &\mid \mathcal{A} \gets R_{q}^{m \times k}; \, y \gets A(\mathcal{A}))     
\end{align*}

\end{definition}

\begin{definition}[Differential Privacy~\cite{abadi2016deep}]\label{def:dp}
A randomized mechanism $\mathcal{M}$: $\mathcal{X} \to \mathcal{Y}$ satisfies $(\epsilon, \delta)$-DP, if for any two adjacent datasets $X, X' \in \mathcal{X}$ that differ in only a single data element and for any subset of output $Y \subseteq \mathcal{Y}$, $\mathbb{P}(M(X) \in Y) \leq exp(\epsilon) \cdot \mathbb{P}(M(X') \in Y) + \delta$ holds. 

\end{definition}

The parameter $(\epsilon, \delta)$ is often called the \textit{privacy budget}. Specifically, $\epsilon$ represents the privacy guarantee: a lower $\epsilon$ corresponds to a higher level of privacy; and $\delta$ indicates the probability that the upper-bound does not hold.

\section{Models}

In this section, we first introduce the system model for \sys, then define its threat and security models. 



\subsection{System Model}

\begin{figure}
  \centering
  \includegraphics[width=1\linewidth]{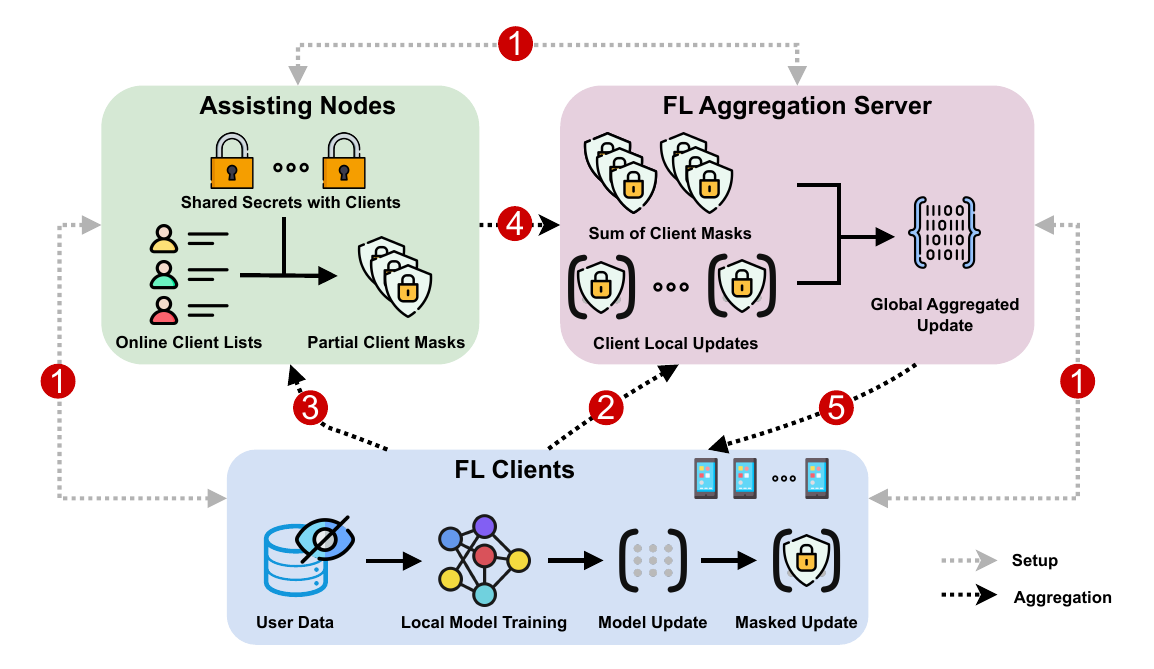}
  \caption{\sys's System Model}
  \label{fig:sys_model}
\end{figure}

We define a system model similar to models defined by previous work~\cite{behnia2023efficient,flamingo} (see  Fig.~\ref{fig:sys_model}). 
Our system consists of three types of participants: an aggregation server, a set of $n$ users with their local dataset collaborating to train a central model, and $k$ assisting nodes, which assist the aggregation server in unmasking the final model without leaking any individual gradient.
\sys offers a one-time setup (\ding{182}) for $T$ training iterations. After the setup,  each client trains the model on their local data, masks the local updates, and sends the masked updates to the aggregation server (\ding{183}) with a simple participation message broadcast to all $k$ assisting nodes (\ding{184}). After receiving all the masked updates, the server receives the aggregated masking terms from the $k$ assisting nodes (\ding{185}). It then aggregates all the provided information to obtain the unmasked global model (\ding{186}). 
To address the different privacy requirements, \sys~seamlessly adopts DP in different stages of training, depending on the target privacy requirements. We discuss each variation in Section~\ref{sec:sys}.

\subsection{Threat Model}~\label{sec:threat_model}

Attacks on FL systems predominantly aim to compromise the integrity and confidentiality of these 
systems~\cite{bonawitz2017practical,bell2020secure,flamingo}. These include traditional Man-in-the-middle (MITM) attacks and emerging attacks targeting the model to compromise the privacy of the training data.
Following ~\cite{flamingo}, we consider a malicious MITM attacker. This is the strongest adversary in the context of FL~\cite{flamingo,microfedml,bell2020secure}. More specifically, aside from the ability to analyze communication, such an adversary can actively force users to drop out, and drop or replace messages, 
compromising up to $(k-1)$ assisting nodes.
The adversary is assumed to have access to quantum computers capable of breaking conventional cryptographic problems. 
Consistent with ~\cite{flamingo}, we consider a secure and authenticated channel between the entities (e.g., via~\cite{paquin2020benchmarking}). For simplicity and following prior privacy-preserving approaches for the FL setting~\cite{flamingo,microfedml,behnia2023efficient}, we assume that the compromised parties will follow the protocol.

Attacks targeting the model to compromise data privacy, irrespective of the adversarial method utilized—whether black-box or white-box—can be classified, based on severity, into membership inference, model inversion, and training data extraction attacks. Such adversaries can attack user gradients and intermediate models during training and the final model after deployment~\cite{carlini2021extracting} by compromising different entities involved in the protocol. 
Unlike traditional centralized training, where an adversary typically has access only to the final model, FL involves multiple participants who may have different objectives and varying levels of trust. Since training occurs in a distributed setting, privacy concerns are not limited to the deployed model but also include the entire training process. As highlighted by~\cite{hayes2023bounding}, it is crucial to protect user data both during training and after deployment. To this end, we propose a full-stack threat model that addresses privacy risks at every stage of the training pipeline and after deployment. This includes attacks on user updates sent to the server, the intermediate global models shared with clients, and the final deployed model.
Note that extensive research has been conducted on model correctness and poisoning attacks~\cite{behnia2023efficient}, which can be utilized alongside our method to enhance client data privacy protection, hence outside the scope of our work.
Therefore, we define three threat models, categorized based on the adversary's capabilities and access to the model in different stages (i.e., user gradients, intermediate or final models). 
Following the above attack categories, \emph{the adversary succeeds if it can infer any information about the training data} (i.e., the membership inference attack).

\begin{itemize}
    \item \textbf{Threat Model 1:} TM1 considers an adversary targeting user gradients during the training phase to undermine the privacy of the honest users' data. In TM1, the adversary captures a compromised aggregation server that aims to undermine the clients' data privacy through its access to their masked gradient. A real-world example is a cloud provider running the federated learning server which, despite performing aggregation as expected, attempts to reconstruct sensitive information (e.g., handwritten digits or medical conditions) from user-submitted updates using gradient inversion attacks.
     \item \textbf{Threat Model 2:} TM2 considers an adversary targeting the intermediate model during the training phase. This threat model accounts for a compromised aggregation server or a subset of the clients (or both) that aim to undermine the clients' data privacy through their access to the intermediate model computed and distributed by the aggregation server at the end of each iteration. A practical scenario is a group of colluding clients in a cross-silo FL setup (e.g., hospitals sharing models for disease prediction) using model update differences across rounds to infer training data from other participants.

     \item \textbf{Threat Model 3:} TM3 considers privacy attacks after the training phase, once the model is deployed. As stated above, these attacks can occur via white-box and black-box access. 
     A real-world instance is a deployed language model accessed via an API (black-box) or downloaded in its entirety (white-box), where adversaries attempt membership inference or data reconstruction attacks to determine whether specific records (e.g., patient names or user queries) were part of the training data.

\end{itemize}
A protocol is considered to provide \emph{full-stack} privacy if it effectively safeguards user data across all the above threat models, ensuring comprehensive protection throughout the entire lifecycle of the model.
 


\subsection{Security Models}
In the following, we define the security models for our primitives. 
The standard security notion for a digital signature scheme is existential unforgeability against adaptive chosen message attacks (EU-CMA), defined below.

\begin{definition}[EU-CMA]\label{def:EUCMA}
Existential Unforgeability under Chosen Message Attack (EU-CMA) experiment $Expt^{\text{EU-CMA}}_{\Pi,A}$  for  a signature scheme $\Pi$:$(\sgnkeygen$, $\sgnsign$, $\sgnverify)$ with an adversary $A$ is defined as follows. 

\begin{itemize}
    \item $(\sk,\pk)\gets\Pi.\sgnkeygen(1^\kappa)$
    \item $(m^*,\sigma^*)\gets A^{\Pi.\sgnsign(\cdot)}(\pk)$
    \item If $1\gets\Pi.\sgnverify(\cdot)$ and $m^*$ was never queried to $\sgnverify(\cdot)$, return `success' otherwise, output `$\bot$'. 
    
\end{itemize}

\end{definition}
The main security notion for a key encapsulation mechanism is the indistinguishability of the chosen ciphertext attack (IND-CCA), which enables the adversary to access the decapsulation mechanism. 
\begin{definition}[IND-CCA]\label{def:INDCCA}
IND-CCA property of a key encapsulation scheme  $\kem=\{\sgnkeygen, \kemenc, \kemdec\}$ with key space $\mathcal{K}$ under chosen ciphertext attack (IND-CCA) experiment $Expt^{\text{IND-CCA}}_{\kem,A}$ with an adversary $A$ is defined as follows. 

\begin{itemize}
    \item $(\sk,\pk)\gets\Pi.\sgnkeygen(1^\kappa)$
     \item $b\gets \{0,1\},(c_{x},x_0)\gets\Pi.\kemenc(\pk)$
     \item $x_1\gets\mathcal{K}$
    \item $b'\gets A^{\Pi.\kemenc(\cdot),\Pi.\kemdec(\cdot)}(\pk,c_{x},x_b)$

\end{itemize}
The advantage of the adversary in the above experiment is defined as $\Pr[b=b']\leq \frac{1}{2} +\varepsilon$ for a negligible $\varepsilon$. 
\end{definition}

\section{Post-Quantum Federated Learning}\label{sec:sys}
Given a secure public key infrastructure, \sys~only requires a single setup round to train a central model (via multiple rounds). 
\sys~provides security against quantum adversaries while incurring the least computational overhead compared to its counterparts with conventional (non-quantum) security. 
Below, we present the high-level ideas of \sys,  followed by a detailed protocol.

\begin{figure*}[!ht]
  \centering
  \includegraphics[width=1\linewidth]{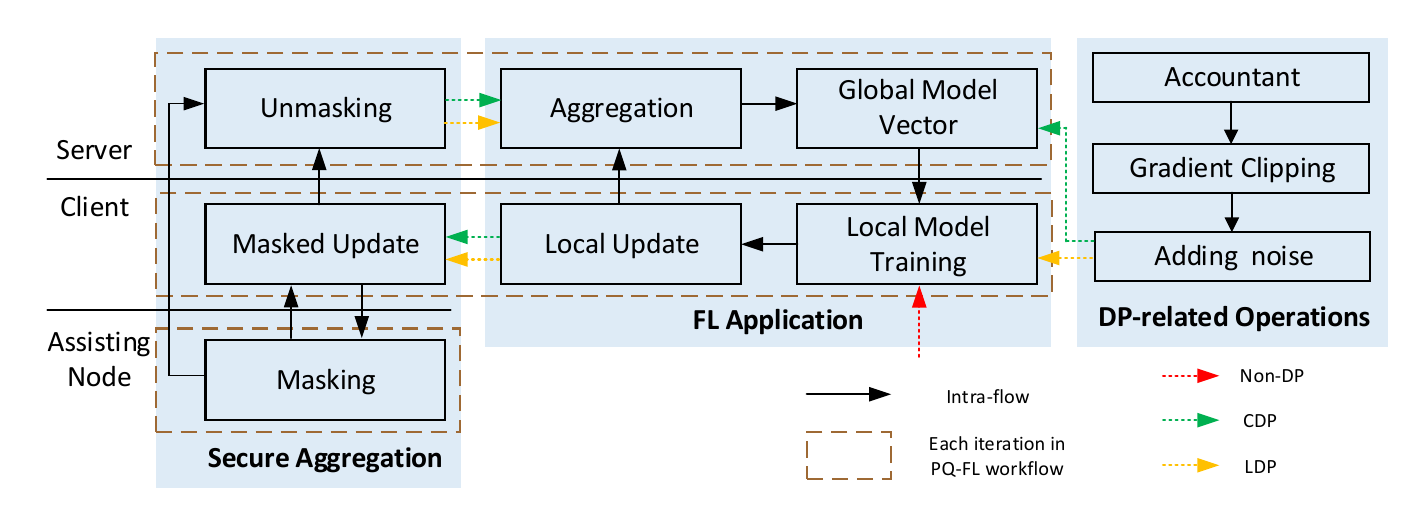}
  \caption{An overview of \sys's architecture and how it fits in the existing FL workflow.}
  \label{fig:workflow}
\end{figure*}


\subsection{High-level Idea}
The design of \sys~is based on two main observations: (1) The main overhead in privacy-preserving FL arises from the underlying cryptographic operations. This is further exacerbated when post-quantum security is considered. (2) The training takes place over a small number of iterations, typically between 20 and 50. 
With these insights, we design a new, highly efficient privacy-preserving FL framework by devising a series of precomputation methods to reduce the cryptographic overhead while attaining post-quantum security. 
%
%
Following the recent development in practical secure aggregation protocols 
(e.g.,~\cite{flamingo,behnia2023efficient}), \sys~assumes a set of assisting nodes to assist the aggregation server in unmasking the intermediate/final model. 
In addition, to ensure privacy against different threat models, we integrate \sys with different DP methods.
Fig. ~\ref{fig:workflow}  shows the high-level architecture of \sys and its integration with the FL training process. 

\subsection{FL Secure Aggregation with Post Quantum Security.}

\begin{algorithm*}[t]\caption{\sys~Setup}\label{alg:sysSetup}
\small
\raggedright{\textbf{Input:} All parties are provided with the security parameter $\kappa $, the number of assisting nodes $k$, a key encapsulation protocol $\kem=\{\sgnkeygen, \kemenc, \kemdec\}$,
 a digital signature scheme with precomputation $\Pi$: $(\sgnkeygen,\precomputesgn,\presgnsign,\sgnverify)$, instantiated using the security parameter $\kappa$. 

\textbf{Output:} All the entities generate their key exchange, signature  key pairs, as well as precomputed signing parameter list $\mathcal{LS}$.

All the users receive the public keys of all the  assisting nodes $\langle(\pk^{\node_1}_\Pi, \dots, \pk^{\node_k}_\Pi),(\pk^{\node_1}_\kem, \dots, \pk^{\node_k}_\kem)\rangle$, and shared secret seed $\x^{\user_i}_{\node_j}$ for each assisting node $\node_j$.
All the assisting nodes receive the public keys of all the  participating users $(\pk^{\user_1}_\Pi, \dots, \pk^{\user_n}_\Pi)$, and a shared secret seed $\x^{\node_j}_{\user_i}$ with each user $\user_i$.  
The aggregation server  \agg~receives the public keys of all the users and  assisting nodes $\langle(\pk^{\user_1}_\Pi, \dots, \pk^{\user_k}_\Pi),(\pk^{\node_1}_\Pi, \dots, \pk^{\node_k}_\Pi)\rangle$.} 
Lastly, all the users and assisting nodes generate the precomputed mask lists $\mathcal{MS}$. \\

\hrulefill \vspace{-1mm}

\textbf{Phase} 1 \textbf{(KeyGen and Advertise)}
\vspace{-2mm}

\hrulefill

All the communications below are conducted via an authenticated channel (similar to \cite{bonawitz2017practical,microfedml})
\begin{algorithmic}[1]

\item Each assisting node $\node_j$ generates its key pair(s) $(\sk^{\node_j}_\Pi,\pk^{\node_j}_\Pi)\gets\Pi.\sgnkeygen(\kappa)$ and$(\sk^{\node_j}_\kem,\pk^{\node_j}_\kem)\gets\kem.\kagen(\kappa)$  and sends $(\pk^{\node_j}_\Pi, \pk^{\node_j}_\kem)$  to the $n$ users and sends $\pk^{\node_j}_\Pi$ the aggregation server.  

\item Upon receiving $\pk^{\node_j}_\kem$ from $\node_j$,  each user $\user_i$ computes and encapsulate a shared secret by $(\x_{\node_j}^{\user_i}, c_{\node_j}^{\user_i})\gets\kem.\kemenc(\pk_\kem^{\node_j})$. It then generates its key pair(s) $(\sk^{\user_i}_\Pi,\pk^{\user_i}_\Pi)\gets\Pi.\sgnkeygen(\kappa)$  and sends  $\pk^{\user_i}_\Pi$ and   $c_{\node_j}^{\user_i}$ to $\node_j$ to the $k$ assisting nodes and $\pk^{\user_i}_\Pi$ to the aggregation server.





\item 
Upon receiving $c_{\node_j}^{\user_i}$ from the the user $\user_i$, the assisting node $\node_j$
computes and retrieves the shared secret  
$\x_{\node_j}^{\user_i} \gets \kem.\kemdec(c_{\node_j}^{\user_i}, \sk_\kem^{\node_j})$

\item The aggregation server \agg~generates a key pair $(\sk^\mathtt{S}_\Pi,\pk^\mathtt{S}_\Pi)\gets\Pi.\sgnkeygen(1^\kappa)$ and sends $\pk^\mathtt{S}_\Pi$ to the users.

\item All the entities perform the precomputation to get a list $\mathcal{LS}$ with $N$ groups of pre-computated parameters for optimizing further signing operations: $\mathcal{LS}_{\node_j} \gets \precomputesgn(\sk_{\node_j}, N)$, $\mathcal{LS}_{\user_i} \gets \precomputesgn(\sk_{\user_i}, N)$, $\mathcal{LS}_{\agg} \gets \precomputesgn(\sk_\agg, N)$.

\item Each user $\user_i$ computes a list $\mathcal{MS}_{\user_i}^{\node_j}$ for each assisting node $\node_j$, with $T$ groups of masks: $\mathcal{MS}_{\user_i}^{\node_j}$[$t$] = $\prf(\x_{\user_i}^{\node_j},t)$ for $t\in[1,\dots,T]$. Similarly, each assisting node $\node_j$ computes a list $\mathcal{MS}_{\node_j}^{\user_i}$: $\mathcal{MS}_{\node_j}^{\user_i}$[$t$] = $\prf(\x_{\node_j}^{\user_i},t)$ for $t\in[1,\dots,T]$.

\end{algorithmic}

\end{algorithm*}

\begin{algorithm*}[!t]\caption{\sys~Aggregation}\label{alg:sysAgg}
\small
\raggedright{\textbf{Input:} The iteration $t$, a user calculated model update $\vec{w}^{\user_i}_t$, a keyed pseudorandom function $\prf $, a digital signature scheme $\Pi$,   the secret seeds $\x$ (shared between the users and the nodes), the signature key pair of the users and the assisting nodes, as well as the precomputed signing parameter list $\mathcal{LS}$ of all entities and the precomputed mask list $\mathcal{MS}$ of the users and the assisting nodes.

\textbf{Output:} The final model update $\vec{w}_t \in \mathcal{M}^d$. \\

\hrulefill \vspace{-1mm}

\textbf{Phase} 1 \textbf{(Masking Updates)}}

\vspace{-2mm}

\hrulefill

\begin{algorithmic}[1]

\State The user compute the local  update $\Vec{w}^{\user_i}_{t}$ for iteration $t$
\State The user $\user_i$ 
first gets the mask $\mathcal{MS}_{\user_i}^{\node_j}$[$t$] for iteration $t$, and then computes $\Vec{a}^{\user_i}_{t}=\sum_{j=1}^k$$\mathcal{MS}_{\user_i}^{\node_j}$[$t$] 
and the masked update $\vec{y}^{\user_i}_{t}=\vec{w}^{\user_i}_{t}+\Vec{a}^{\user_i}_{t}$. 
%


\State The user sets $m=(t,\vec{y}^{\user_i}_{t})$ and $m'=(t)$, computes   signatures $\sigma^{\user_i}_{\agg} \gets \Pi.\presgnsign(\sk^{\user_i}_\Pi,m, \mathcal{LS}_{\user_i})$ and $\sigma^{\user_i}_{\node_j} \gets \Pi.\presgnsign(\sk^{\user_i}_\Pi,m', \mathcal{LS}_{\user_i})$ and sends $(m,\sigma^{\user_i}_{\agg})$ and $(m',\sigma^{\user_i}_{\node_j})$ to the aggregation server and the assisting node $\node_j $ (for $j\in[1,\dots,k]$), respectively. 



\end{algorithmic}

 \hrulefill \vspace{-1mm}

\textbf{Phase} 2 \textbf{(Aggregate Updates)}
\vspace{-2mm}

\hrulefill

\begin{algorithmic}[1]
\State Upon receiving $(m',\sigma^{\user_i}_{\node_j})$
from all the  users in the system, $\node_j$ checks if  $ \Pi.\sgnverify(\pk_\Pi^{\user_i},m',\sigma^{\user_i}_{\node_j} ) \stackrel{?}{=} 1$ holds, it adds the user to its user list $\mathcal{L}_{j,t}$. 

\State Each assisting node checks if  {$|\mathcal{L}_{j,t}|\geq \alpha\ppp_H$}, 
it gets $\mathcal{MS}_{\node_j}^{\user_i}$[$t$], computes $\Vec{a}^{\node_j}_{t}=\sum_{i=1}^{|\mathcal{L}_{j,t}|}$$\mathcal{MS}_{\node_j}^{\user_i}$[$t$],
sets $m^{\prime\prime}=(t,|\mathcal{L}_{j,t}|,\Vec{a}^{\node_j}_t)$ and $\sigma^{\node_i}_{\agg} \gets \Pi.\presgnsign(\sk^{\node_j}_\Pi,m'', \mathcal{LS}_{\node_j})$  and sends $(m'',\sigma^{\node_j}_{\agg})$ to the aggregation server \agg.

\State Upon receiving $(m,\sigma^{\user_i}_{\agg})$, \agg~first checks if 
$ \Pi.\sgnverify(\pk_\Pi^{\user_i},m,\sigma^{\user_i}_{\agg} )\stackrel{?}{=}  1$ holds, it adds $\user_i$ to its user list $\mathcal{L}_{\agg,t}$.

\State Next, \agg~checks if   $  \Pi.\sgnverify(\pk_\Pi^{\node_j},m'',\sigma^{\node_j}_{\agg} )\stackrel{?}{=}  1$ holds   for $j\in[1,\dots,k]$. It then checks if  $|\mathcal{L}_{\agg,t}|= |\mathcal{L}_{1,t}| =\dots = |\mathcal{L}_{k,t}|$, does not hold it aborts. 

\State \agg~computes and broadcasts the final update as $\vec{w}_t=\sum_{i=1}^{|\mathcal{L}_{\agg,t}|}\vec{y}^{\user_i}_{t}-\sum_{j=1}^k\Vec{a}^{\node_j}_t$.

 \end{algorithmic}
\end{algorithm*}

\setlength{\floatsep}{0.1cm}

\subsubsection{Setup Phase}
The setup phase is a one-time process that can be performed, at least partially, offline. 
This phase consists of a single round, namely, \textit{KeyGen and Advertise}. 
During this round, all users ${P_1, ..., P_n}$ and assisting nodes ${A_1, ..., A_k}$ are initialized with the two pairs of public keys, i.e., ($\sk_{\kem}$, $\pk_{\kem}$) for the key exchange scheme and ($\sk_{\Pi}$, $\pk_{\Pi}$) for the digital signature scheme. 
Then, a public-key exchange procedure is performed as follows: 
(1) the users receive copies of the key exchanging public keys from each assisting node $\pk_{\kem}^{A_j}$;
(2) the assisting nodes receive copies of the user signing public keys $\pk_{\Pi}^{P_i}$;
and (3) the aggregation server receives copies of the signing public keys from the assisting nodes $\pk_{\Pi}^{A_j}$ and the users $\pk_{\Pi}^{P_i}$. 
After that, each user computes a shared secret $x_{P_i}^{A_j}$ for each assisting node $A_j$ with $\sk_{\kem}^{P_i}$ and $\pk_{\kem}^{A_j}$ via the key exchange scheme $\kem.\kemenc(.)$ and securely sends it to the corresponding assisting node $A_j$.
The assisting node $A_j$ can recover the shared secret using its private key $\sk_{\kem}^{A_j}$ via $\kem.\kemdec(.)$.


\subsubsection{Aggregation Phase}
The aggregation phase (Algorithm 2) comprises two rounds: \textit{Masking Updates} and \textit{Aggregate Updates}.
The aggregation phase is repeated at run-time with several iterations based on the FL training requirements.

Specifically, at $t$-th iteration, in the first round (\textit{Masking Updates}), each user computes a masking vector $a_t^{P_i}$ using the shared secrets $x_{P_i}^{A_j}$ for $j = {1, ..., k}$ to mask its user gradient $w_t^{P_i}$. The user then sends the masked gradient $y_t^{P_i}$ and a participation message (e.g., the iteration number) to the aggregation server and the $k$ assisting nodes, respectively. All the messages are signed using $\Pi.\sgnsign(.)$.

In the second round (\textit{Aggregate Updates}), upon receiving (and verifying) the participation message (and signature), each $A_j$ adds the user to the list $\mathcal{L}{j,t}$. For all users in $\mathcal{L}{j,t}$, $A_j$ computes the aggregation of their masking vectors $a_t^{A_j}$ using the shared keys $x_{P_i}^{A_j}$ and sends $a_t^{A_j}$ to the aggregation server, in which the messages are signed by assisting nodes using $\Pi.\sgnsign(.)$.
Next, when receiving (and verifying) the participation message (and signature), the aggregation server $\mathcal{S}$ adds the user to the list $\mathcal{L}{s,t}$ and then verifies that all user lists from the assisting nodes and the aggregation server are identical. 
For all users in $\mathcal{L}{s,t}$, the server uses the masked updates ($y_t^{P_i}$ and the aggregated masking terms $a_t^{A_j}$) provided by the assisting nodes to efficiently compute 
the aggregated intermediate model.

\subsubsection{Precomputed Masks}
To avoid plain gradient transmission, \sys generates a mask for each client local update and only transmits the masked update.
On the user side, a mask is generated with the shared secret of the user and each assisting node $x_{P_i}^{A_j}$ and the iteration number $t$.
As the shared secrets are calculated in the setup phase and the iteration number is iterated in order, there is no any runtime information (i.e., local updates) involved.
Therefore, we can precompute $T$ masks for $T$ iterations before the runtime aggregation.

\subsection{Resiliency Against User and Assisting Node Dropouts}
The design of \sys offers strong resilience against both user and assisting node dropouts. Specifically, \sys is a one-round protocol. Unlike existing approaches~\cite{bonawitz2017practical, bell2020secure,microfedml}, user dropouts do not introduce any additional overhead for the remaining participants. As in standard non-private FL, user dropouts may affect model accuracy but do not impact the overall protocol execution. Therefore, \sys does not introduce any additional constraints in this regard.
Moreover, the novel design of \sys enables the use of a rotating set of assisting nodes. In scenarios where an assisting node has an unreliable connection or is at risk of dropping out, a simple secret sharing scheme can be used to distribute the node’s secret among the other assisting 
nodes~\cite{shamirSecret}. This schema allows the online assisting parties to collaboratively reconstruct the required masking terms without imposing any additional burden on the regular users.

\subsection{Integrating Differential Privacy}\label{sec:pqfl:dp}
During and after the FL aggregation, attackers under different threat models could exploit the client gradients, intermediate models, and final models to conduct black- and white-box-based model inversion attacks, as outlined in Section~\ref{sec:threat_model}. These attacks could allow one to infer information about or reconstruct the model's training data, making secure aggregation alone insufficient to fully protect user training data.

To tackle that issue, \sys integrates DP techniques with secure aggregation, adding noises to the final, intermediate, and local models during training.
Depending on the threat models, in which adversaries with different abilities are involved, we employ two DP methods--\textbf{Local Differential Privacy (LDP)} and \textbf{Central Differential Privacy (CDP)}-- individually or in combination to safeguard the client gradient privacy, intermediate model privacy, and final model privacy.
Specifically, LDP is used to protect client gradient updates. Each client applies DP noise during local model training at every iteration, preventing user training data from being exposed to a malicious server or other clients. Note that clients may use different privacy budgets ($\varepsilon$), with the overall privacy guarantee determined by the largest $\varepsilon$ among all clients.
CDP protects against model inversion attacks on aggregated models. The CDP server adds DP noise to the aggregated model at each training iteration, ensuring that malicious entities cannot infer the individual data of honest clients from the aggregated models.

\sys adopts tailored DP strategies with the two DP methods based on the specific threat model.
Under \textbf{TM1}, \sys applies LDP (as well as the secure aggregation) to protect the client gradients. By adding noise to client gradients during local client training, a compromised server will be unable to infer the exact training data after receiving the client gradients.
Under \textbf{TM2} and \textbf{TM3}, \sys utilize LDP and CDP to protect the intermediate models and final models against compromised clients and compromised server.
Noise is added to the client gradients on the client side and to the aggregated model on the server side, ensuring robust protection under these threat models.

\subsection{Efficient Post-Quantum Signature}~\label{sec:pqfl:sign}



\begin{algorithm}[t]




\raggedright{\textbf{$\mathcal{LS}\gets \precomputesgn(\sk_\Pi, N)$}}
\vspace{-2mm}

\hrulefill 
\begin{algorithmic}[1]

\State $\mathcal{LS} := \emptyset$ 
\State \textbf{while} \text{Length of} $\mathcal{LS} < N$ \textbf{do} 
\State \hspace{\algorithmicindent} $\mathbf{y} \leftarrow \mathbb{S}_{\gamma_1 - 1}^{\ell}$ 
\State \hspace{\algorithmicindent} $\mathbf{w}_1 := \mathsf{HighBits}(\mathbf{A}\mathbf{y}, 2\gamma_2)$ 
\State \hspace{\algorithmicindent} \text{Add} ($\mathbf{y}, \mathbf{w}_1$) \text{into} $\mathcal{LS}$ 
\State \textbf{return} $\mathcal{LS}$ 
\end{algorithmic}

\hrulefill \vspace{-1mm}

\textbf{$\sigma \gets \presgnsign(\sk_\Pi, m, \mathcal{LS})$}
\vspace{-2mm}

\hrulefill

\begin{algorithmic}[1]
\State \textbf{for each} ($\mathbf{y}, \mathbf{w}_1$) \textbf{in} $\mathcal{LS}$ 
\State \hspace{\algorithmicindent} $c \in \mathbb{B}_{60} := \mathcal{H}(m \parallel \mathbf{w}_1)$ 
\State \hspace{\algorithmicindent} $\mathbf{z} := \mathbf{y} + c\mathbf{s}_1$ 
\State \hspace{\algorithmicindent} \textbf{if} $\|\mathbf{z}\|_\infty \geq \gamma_1 - \beta$ 
\State \hspace{\algorithmicindent} \hspace{1em} \textbf{or} $\|\mathsf{LowBits}(\mathbf{A}\mathbf{y} - c\mathbf{s}_2, 2\gamma_2)\|_\infty \geq \gamma_2 - \beta$, 
\State \hspace{\algorithmicindent} \hspace{1em} \textbf{then} \text{Remove} ($\mathbf{y}, \mathbf{w}_1$) \text{from} $\mathcal{LS}$ 
\State \hspace{\algorithmicindent} \hspace{1em} \textbf{return} $\sigma = (\mathbf{z}, c)$ 
\State \textbf{return} \underline{$\sgnsign(\sk_\Pi, m)$} 
\end{algorithmic}
\caption{Precomputation Algorithm for Dilithium}\label{alg:dili}
\end{algorithm}

Digital signatures are essential for ensuring the authenticity and integrity of transmitted messages. While post-quantum signatures promise long-term security against quantum adversaries, they are often more resource-intensive than their classical counterparts. This overhead becomes significant when these schemes are deployed on resource and battery-constrained devices like cellular phones. 
Several algorithms have been proposed in the NIST post-quantum standards, such as Dilithium~\cite{dilithium} and SPHINCS+~\cite{bernstein2015sphincs}. We select Dilithium as the most efficient option~\cite{li2024enhancing}, and optimize it for deployment in \sys. In the following, we describe Dilithium and introduce a precomputation method to significantly improve its computational overhead in the FL settings. 

Dilithium is based on Fiat-Shamir with abort paradigm, and its security relies on the hardness of the modulo-SIS and modulo-LWE problems. 
In the standard Dilithium signature without any optimization (as Algorithm~\ref{alg:dilithium} in Appendix~\ref{sec:app:algo}), 
key generation involves selecting matrix $\textbf{A}$ of dimension $k\times l$ and computing an MLWE public key $\textbf{t}=\textbf{A}s_1+s_2$ where $s_1$ and $s_2$ are sampled from the key space  $\mathcal{S}_\eta$ with small coefficient of size at most $\eta$. The $\mathsf{HighBits}(\cdot)$ and 
$\mathsf{LowBits}(\cdot)$ methods simply select the high-order and low-order bits of the coefficients in their vector, respectively. 

Given the Fiat-Shamir paradigm, in lattice-base settings, since we are solving for a variant of the SIS problem, the signature $\mathbf{z}$ should be small. To ensure this, an eligible signature should meet two conditions: (i) For security, it is crucial that the masking term $\mathbf{y}$ is completely hiding $c\mathbf{s}_1$;  the first rejection condition (Step 7) addresses this concern. (ii) The second rejection condition is required for both the scheme's correctness and 
security~\cite{behnia2021removing} and ensures that the same $c$ is recovered in the verification algorithm. 
In practice, this could require repeating the signing algorithm $~7$ times before finding an eligible signature~\cite{dilithium,behnia2021removing}. This can have a significant resource and computation overhead for constrained devices.

To address such an issue, we propose a precomputation algorithm for Dilithium by leveraging the fact that FL takes a small number of training iterations (in the order of hundreds). The high-level idea of the precomputation algorithm, presented in 
Algorithm~\ref{alg:dili}, is to store a list of masking values $\mathbf{y}$ and their corresponding commitment values $\mathbf{Ay}$. Therefore, if the selected masking term does not pass the rejection conditions during the signing algorithm, another value and its corresponding commitment can be selected from the list. Following the one-time nature of the masking term, it will be removed from the list once an eligible masking term is selected. Detailed algorithms are provided in Appendix~\ref{sec:app:algo}.

\subsection{Security Analysis}

\begin{theorem}\label{thm:malicious}
    The \sys~protocol presented in Algorithms \ref{alg:sysSetup} and \ref{alg:sysAgg}, running with $n$ parties $\{ \user_1,\dots, \user_n\}$, $k$ assisting nodes $\{\node_1,\dots,\node_k\}$,  and an aggregation server \agg~provides  post-quantum privacy against  a malicious adversary $A \in {\textbf{A1, A2, A3}}$~which controls \agg~and  $1-\alpha$ fraction of  users and $k-1$  assisting nodes with the offline rate $(1-\delta)\geq\alpha$.   
\end{theorem}

\begin{proof}

Following~\cite{behnia2023efficient}, we prove the above theorem by the standard hybrid argument. The proof relies on the assumption of the simulator \Sim~that controls the environment. $\ppp_h$ and $\ppp_C$ define the set of honest and corrupt users, respectively.

\smallskip
\noindent

\emph{Setup phase:}
\begin{enumerate}
    \item Each honest user $\user_i$ and assisting node $\node_j$ follows the protocol in Algorithm \ref{alg:sysSetup}.
    \item For each  corrupt user $\user_{i'}\in\ppp_{C}$ and honest assisting node  $\node_j\in\aaa_H$,  compute  and store  $(\x_{\node_j}^{\user_{i'}}, c_{\node_j}^{\user_{i'}})\gets\kem.\kemenc(\pk_\kem^{\node_j})$.

    \item For each  honest user $\user_{i}\in\ppp_H$ and corrupt assisting node  $\node_{j'}\in\aaa_C$,  compute and store   $(\x_{\node_{j'}}^{\user_{i}}, c_{\node_{j'}}^{\user_{i}})\gets\kem.\kemenc(\pk_\kem^{\node_{j'}})$.
    \item For each honest pair of user $\user_i$ and honest assisting node $\node_j$, \Sim~ picks $\x_r\Ra\mathcal{K}_\Sigma$ and sets $   \x^{\user_i}_{\node_j}= \x_r$.
\end{enumerate}

\smallskip
\noindent
\emph{Aggregation phase:}
  
\begin{enumerate}
    \item In each iteration $t$ of the protocol, each honest user picks a random $\Vec{y'}_t^{\user_i}$. It then computes two signatures  $\sigma^{\user_i}_{\agg} \gets \Pi.\sgnsign(\sk^{\user_i}_\Pi,m)$ and $\sigma^{\user_i}_{\node_j} \gets \Pi.\sgnsign(\sk^{\user_i}_\Pi,m')$  and sends $(m{,\sigma^{\user_i}_{\agg}})$ and $(m'{,\sigma^{\user_i}_{\node_j}})$ to the aggregation server and the assisting node $\node_j $ (for $j\in[1,\dots,k]$), respectively. 
    \item The aggregation server  \agg~first adds the user $\user_i$ to the list $\Ulist_{\agg,t}$ and then calls the $\alpha$-summation ideal functionality $\mathcal{F}_{\Vec{x},\alpha}(\Ulist_{\agg,t}\backslash \mathcal{P}_{C})$ (where $\mathcal{P}_{C}$ is the set of corrupt users) to get $\Vec{w}_t$.

\item Next, the simulator samples $\Vec{w'}_t^{\user_i} \Ra \mathcal{M}^d$ for all $\user_i \in \Ulist_{\agg,t}\backslash \mathcal{P}_{C} $ such that $\Vec{w}_t=\sum_{i\in\mathcal{P}_{C}}\Vec{w'}_t^{\user_i} $ and computes $\Vec{a'}_{t}^{\user_i}=\Vec{y'}_t^{\user_i} - \Vec{w'}_t^{\user_i}$.
\Sim~sets $\{\Vec{a'}_{\node_j,t}^{\user_i}\}_{\node_j \in \aaa_H} = \{\RO(\x_{\node_j}^{\user_i},t)\}_{\node_j\in\aaa_H}$ such that $\Vec{a'}_{t}^{\user_i}=\sum_{\node_j \in \aaa_H} \{\Vec{a'}_{\node_j,t}^{\user_i}\}$, and for each $\node_j\in\aaa_H$, it computes  $ \Vec{a'}_{\user_i,t}^{\node_j}=\sum_{\user_i\in\ppp_H}\Vec{a'}_{j,t}^{\user_i}$. 

\end{enumerate}

The hybrids are provided below. We note that each hybrid represents a view of the system seen by the adversary. The proof relies on the indistinguishability of each hybrid.  The hybrids are constructed by the simulator \Sim. 

\begin{description} 

\item [$\mathtt{Hyb0}$] The random variable is identical to the real execution of the protocol (i.e., \real). 
\item [$\mathtt{Hyb1}$] Now, \Sim~that knows the secrets of all honest entities are introduced in this hybrid. The distribution of this hybrid remains identical to the one above. 


\item [$\mathtt{Hyb2}$] This hybrid replaces the shared keys between the users and assisting nodes with a random shared secret key sampled from  $\mathcal{K}_\mathtt{\Sigma}$.  The \indist of this hybrid is delivered by the security of the key encapsulation mechanism (Definition~\ref{def:INDCCA}). Our protocol delivers this by the security of the CRYSTALS-Kyber 
algorithm~\cite{bos2018crystals}, which is based on the Module-LWE problem which is post-quantum secure. 


\item [$\mathtt{Hyb3}$] This hybrid starts with each honest user picking a random vector $\Vec{y'}^{\user_i}_t$. The user then computes a signature $\sigma^{\user_i}_{\agg} \gets \Pi.\sgnsign(\sk^{\user_i}_\Pi,<\Vec{y'}^{\user_i}_t,t>)$ and sends $(\Vec{y'}^{\user_i}_t,\sigma^{\user_i}_{\agg})$ to the server.  The \indist of this hybrid is due to the following. 1) Given we need at least one assisting node to remain honest, and since $A$ does not have the secret of honest entities, in \real, the masked gradient vector will have the same distribution as $\Vec{y'}^{\user_i}_t$ and therefore indistinguishable. 2) The  \indist of $\sigma^{\user_i}_{\agg}$ is delivered via the security requirement (Definition \ref{def:EUCMA}) of the underlying signature scheme. 


\item [$\mathtt{Hyb4}$] The ideal functionality $\mathcal{F}_{\Vec{x},\alpha}(\Ulist_{\agg,t}\backslash \mathcal{P}_{C})$ and random oracles are used to substitute the aggregated mask outputted by the honest assisting nodes  $\node_i\in\aaa_H$ with  $ \Vec{a'}_{\user_i,t}^{\node_j}$ (computed in Step 3 of the simulated Aggregation phase presented above). This hybrid outputs $( \Vec{a'}_{\user_i,t}^{\node_j},\sigma^{\node_i}_{\agg})$ where $\sigma^{\node_i}_{\agg} \gets \sgnsign(\sk^{\node_j}_\Pi,m'')$. Given the ideal functionality, random oracles, and digital signatures, and since $A$ does not have the secret of honest entities, the view of this hybrid is \indist with the previous hybrid. 


\item [$\mathtt{Hyb5}$] In the final hybrid, \Sim~outputs the output of the ideal functionality  (Step 2 in the simulated Aggregation phase)  as the global update $\vec{w}_t$. Given the assumption of the fraction of honest users/nodes, the ideal functionality will not output $\bot$ with an overwhelming probability. Thus, this hybrid is \indist from the previous one. 

\end{description}
In the above, we have shown that the view of all the corrupted parties controlled by $A$~is computationally indistinguishable. This ensures that the masked gradients are essentially random values, providing full privacy and thereby protecting the privacy of the clients' dataset in the sense of TM1.   
 
\end{proof}

\begin{corollary}
The protocol presented in Algorithm \ref{alg:sysAgg} protects privacy against a malicious adversary capable of controlling communication and arbitrarily dropping users.
\end{corollary}
\begin{proof}
 The privacy guarantee of \sys~against user dropouts is provided by the $\alpha$-summation ideal functionality (Definition \ref{def:alphasummation}) and enforced by a condition check during the Aggregation phase (Algorithm \ref{alg:sysAgg}). More specifically, to prevent privacy leakages, in the protocol each assisting node checks if at least $\alpha\mathcal{P}_H$ users are participating in the training phase. 
This ensures that the aggregation process continues only when a sufficient number of users are contributing, preventing the adversary from inferring individual data through selective dropouts.
 
\end{proof}

\begin{corollary}
Given a malicious user, the protocol presented in Algorithm~\ref{alg:sysAgg} combined with the CDP method provides $(\epsilon_c,\delta_c)$-privacy in the context of TM2, and TM3.
\end{corollary}

\begin{proof}
Following the proof of 
Theorem~\ref{thm:malicious}, \sys~protects the privacy of user gradients and allows the aggregation server to compute the intermediate (or final) model without revealing any information about the gradients. The Central Differential Privacy (CDP) method is applied on the server side, ensuring that the intermediate (or final) model is $(\epsilon_c, \delta_c)$-protected against malicious clients during training and deployment.

However, it is essential to note that since CDP is implemented on the server, a malicious server can access the plain intermediate (or final) model before CDP is applied. Consequently, while CDP provides $(\epsilon_c, \delta_c)$-privacy from malicious clients during training and deployment, it still assumes an honest aggregation server. $\epsilon_c $ and $\delta_c$ define the central privacy budget and the failure probability computed by a privacy accountant on the server side.  
\end{proof}

\begin{corollary}\label{col:TM2}
Given a malicious aggregation server, the protocol presented in Algorithm \ref{alg:sysAgg} combined with the LDP method provides $(\epsilon_l,\delta_l)$-privacy in the context of TM1, TM2, and TM3.
\end{corollary}

\begin{proof}
    Compared to the previous one, the key distinction in this corollary lies in the use of Local Differential Privacy (LDP) alongside \sys~and the removal of the trusted server assumption. Specifically, in the LDP setting, obscuring user gradients occurs on the client side, ensuring that the intermediate and final models computed by the aggregation server maintain $(\epsilon_l,\delta_l)$-privacy at all times. It is important to note that \sys~fully conceals user gradients. Here, $\epsilon_l$ and $\delta_l$ represent the largest privacy budget and the corresponding probability of failure among the participating clients.
\end{proof}

\section{Performance Evaluation }

\subsection{Analytical Evaluation}\label{sec:analytical}
We present a comparative analysis of the computational performance of \sys and several state-of-the-art protocols in Table~\ref{tab:analytical}. 

\subsubsection{Setup Phase}
The setup phase of \sys involves two runs of the key generation protocol and $k$ runs of the key agreement protocol for the user, or $|U|$ runs for the assisting node. 
Additionally, each user and assisting node perform $T \times k$ and $T \times |U|$ $\PRF$ invocations, respectively, to precompute $T$ groups of masks, along with $N$ $\precomputesgn$ operations to precompute $N$ groups of signing parameters.

\sys offers improved computational efficiency over \flamingo, as it requires fewer participating decryptors—$k$ can be as low as 2—compared to a minimum of 64 in \flamingo~\cite{flamingo}. Moreover, \sys is more efficient than the maliciously-secure version of \micro~\cite{microfedml}, which involves a linear number of operations, including signature verification, key agreement, secret sharing, and symmetric encryption/decryption, leading to higher computational costs.
While \sys incurs additional overhead in precomputing signing parameters and masks compared to \eseafl, this cost is offset by the increased efficiency in the subsequent aggregation phase, with minimal impact on the offline setup phase. 
Additionally, in contrast to the protocol by Bell et al.~\cite{bell2020secure}, which requires the setup phase to be repeated in every aggregation round due to the inability to reuse key material, \sys avoids this repetitive overhead, further enhancing its overall efficiency.

\subsubsection{Aggregation Phase}
During the aggregation phase, \sys further reduces computational overhead. The user only needs to perform $k + 1$ summations and two optimized signature generation operations to achieve  security against malicious adversaries. 
Notably, $\PRF$ computations are not required during aggregation since the masks are precomputed in the Setup phase. Given that $k$ is much smaller than the number of participating users or decryptors, \sys offers significantly higher efficiency compared to alternative protocols.

\begin{table*}[!t]
	\centering
 
	 \caption{Analytical performance of the secure aggregation of \sys~and its counterparts}
	\label{tab:analytical}
 	 \resizebox{1\textwidth}{!}{
	\begin{threeparttable}
	\begin{tabular}{|c|c|c|c||c|c|c|}
			\hline
		\multirow{2}{*}{\textbf{Scheme}} & \multicolumn{3}{c||}{\textbf{Setup Phase}} & \multicolumn{3}{c|}{\textbf{Aggregate Phase} }	\\	 \cline{2-7}
		 & \textbf{User} &
		 \begin{tabular}{@{}c@{}c@{}}
		 \textbf{Assisting/Helper} \\ \textbf{Nodes} \end{tabular}& \textbf{Server} &
		 \textbf{User} &
		 \begin{tabular}{@{}c@{}c@{}}
		 \textbf{Assisting/Helper} \\ \textbf{Nodes} \end{tabular} & \textbf{Server}	\\	\hline

\micro \cite{microfedml} &
		\begin{tabular}{@{}c@{}c@{}} $1_\Kg + 1_\Sgn+|U'-1|(1_\Vfy+$
		\\
		$1_\Ka + 1_{\sshl} + 1_\AEnc+1_\ADec)$
		\end{tabular} & N/A
		& ${|U|_\Vfy }$
		&\begin{tabular}{@{}c@{}c@{}} $(|w|+1)_\Ex+1_\Sgn+$\\$|U-1|(1_\Su+1_\Vfy)+1_\Su$ \end{tabular}& N/A & \begin{tabular}{@{}c@{}c@{}} $1_{\Rsshe} +|w|_\Ex + $\\
		$1_\Mu$\end{tabular}	\\	\hline

	\flamingo	\cite{flamingo} &
		$|D'|_\Vfy$
		& \begin{tabular}{@{}c@{}c@{}}$2_{\PIL} + 3_\Sgn+2_{\sshl} +$	\\ $	|D|(3|D'|+1)_\Vfy + $\\
		$(|D'|\cdot|D|)_\Mu$\end{tabular}
		& $(|D'|\cdot|D|)_\Vfy$ & \begin{tabular}{@{}c@{}c@{}} $2_\PRG+|D'|(1_\Ex+1_\Hash+1_\PRF+$\\$1_\Su+1_\AEnc+1_\Sgn+1_\SEnc)+1_{\sshl}$ \end{tabular}&\begin{tabular}{@{}c@{}c@{}}$1_\Sgn+|U'|_\SDec +$ \\ $(|D|+|U'|)_\Vfy
		$\end{tabular}
		&
		\begin{tabular}{@{}c@{}c@{}} $|D|_{\PIL}+|U'|(1_\Su + $ \\ $1_\Mu +1_\Rssh)+2_\PRG+ $ \\ $ (|U'|+1)_\Su $ \end{tabular}
		\\	\hline

		 \eseafl \cite{behnia2023efficient}&
		 $2_\Kg + k_\Ka$
		 & $2_\Kg + k_\Ka$
		 & $1_\Kg$
		 & $k_\PRF+|k+1|_\Su+2_\Sgn$ &\begin{tabular}{@{}c@{}c@{}} $|U'| (1_\Vfy+1_\Su+ $\\
		 $1_\PRF)+ {1_\Sgn}$\end{tabular}
		 & \begin{tabular}{@{}c@{}c@{}}$(|U'|+|k|)(1_\Vfy + $ \\ $1_\Su)$
		 \end{tabular}\\
		 \hline

		 \sys &
		 $2_\Kg + k_\encap + N_\precomputesgn+(T\times k)_\PRF$
		 & $2_\Kg + |U|_\decap + (T\times|U|)_\PRF$ $+N_\precomputesgn$
		 & $1_\Kg$
		 & $|k+1|_\Su+2_\PSgn$ &\begin{tabular}{@{}c@{}c@{}} $|U'| (1_\Vfy+1_\Su) $\\
		 $+ {1_\PSgn}$\end{tabular}
		 & \begin{tabular}{@{}c@{}c@{}}$(|U'|+|k|)(1_\Vfy + $ \\ $1_\Su)$
		 \end{tabular}\\
		 \hline
		 		 
	\end{tabular}
\begin{tablenotes}[para, flushleft]
	 
	$\Kg$, $\Sgn$, $\Vfy$, and $\Hash$ denote key generation, signature generation, verification, and hash function, respectively.
	$\Ka$ denotes the shared key computation in the key agreement protocol.
	$|D|$ and $|D'|$ denote the number of all decryptors and the threshold of participating decryptors~\cite{flamingo}, respectively.
	$\PIL$ denotes polynomial interpolation of length $L$.
	$\sshl$, $\Rssh$, and $\Rsshe$ denote secret sharing operation to $l$ shares, share reconstruction, and share reconstruction in the exponent, respectively.
	$\Mu$, $\Su$ and $\Ex$ denote multiplication, summation and exponentiation operations, respectively.
	$\PRG$ and $\PRF$ denote pseudorandom generators and pseudorandom functions, respectively.
	$\AEnc$, $\ADec$, $\SEnc$, and $\SDec$ denote asymmetric encryption, decryption, symmetric encryption, and decryption, respectively.
	$|w|$ denotes the number of model parameters.
	$k$ is the set of assisting nodes in \sys.
	\end{tablenotes}
	\end{threeparttable}
		 }
   
\end{table*}

\subsection{Implementation}
We implement \sys with 2.5k lines of code, involving two components, namely the secure aggregation for gradient updates and the model training with multiple DP protections.
Our code is available at \url{https://github.com/kydahe/Beskar}.

\noindent
\textbf{Secure Aggregation.}
\sys builds on ABIDES~\cite{byrd2019abides}, a discrete event simulation framework commonly used in FL research~\cite{flamingo,microfedml}, enabling the simulation of multi-iterative aggregation protocols.
We incorporated two post-quantum algorithms, i.e., the Kyber key encapsulation algorithm for negotiating shared secrets between clients and assisting nodes, and Dilithium signing algorithm for generating and verifying digital signatures during aggregation phase.
These algorithms are implemented based on existing libraries~\cite{kyberpy, dilic}, with Dilithium modified to support signing with precomputation in C. 
For the Pseudo-Random Function (PRF) used in mask generation, we employed ASCON~\cite{dobraunig2021ascon}, known for its lightweight and efficient cryptography, making it suitable for resource-constrained client devices.


\noindent
\textbf{Model Training with Multiple DP Protections.}
We adopt the Flower FL framework~\cite{beutel2020flower}
and integrate our 
DP protections using PyTorch Opacus~\cite{opacus}. In particular, Opacus determines the required noise multiplier  through an iterative procedure that balances privacy and utility. It first converts the training configuration into a total number of steps (via the number of iterations and client fraction fits), then uses a DP accountant (e.g., Rényi differential privacy~\cite{mironov2017renyi}) to assess \(\epsilon\) for a given noise multiplier. The algorithm locates an appropriate noise multiplier by expanding an upper bound until \(\epsilon\) is satisfied, followed by a binary search for finer precision. We implement two different DP methods—tailored to distinct threat models—alongside a baseline FL model without DP protection. Each model is trained five times, and the highest accuracy is reported, accounting for random variations in initialization and privacy mechanisms.



\subsection{Experimental Environment}
The secure aggregation experiments were conducted on an x86\_64 Linux machine with AMD Ryzen Threadripper PRO 5965WX 24-Cores and 256 GB RAM. 
The models (with DP schemes) were trained using a NVIDIA RTX 6000 Ada Generation GPU on another x86\_64 Linux machine with Intel(R) Xeon(R) w7-2475X and 256 GB RAM.



\subsection{Experimental Setup and Evaluation Metrics}
We systematically evaluated \sys with respect to two key dimensions: efficiency  and performance.
Efficiency was measured by empirical time complexity during the setup and aggregation phases, along with corresponding bandwidth usage. 
Performance was assessed based on accuracy using 
five standard vision benchmarks that grow in difficulty and size: MNIST~\cite{xiao2017fashion}, EMNIST~\cite{cohen2017emnist}, CIFAR-10~\cite{krizhevsky2009learning}, CIFAR-100~\cite{krizhevsky2009learning}, and CHMNIST~\cite{kather2016multi}. We pair each dataset with a representative architecture of ascending capacity: a lightweight MLP for MNIST and EMNIST, ResNet-18 for CIFAR datasets, and a larger AlexNet-style network for CHMNIST.
In the following, we describe the three empirical metrics employed in our experiments.


\noindent
\textbf{Metrics for Efficiency:} We use two metrics, i.e., computation time and communication bandwidth, for efficiency measurement.
Comparative analysis was conducted using \flamingo\cite{flamingo}, \pqsa\cite{yang2022post}, \eseafl\cite{behnia2023efficient}, and \micro\cite{microfedml}, where \micro includes two variants, \microfirst and \microsecond.
\begin{itemize}
\item \textit{Computation Time.} To assess the efficiency of our secure aggregation method with post-quantum protection, we measured the time required for computational operations such as key encapsulation, mask generation, signing, and verification. The number of clients ($n$) was varied from 200 to 1000, with three assisting nodes ($k=3$) and one aggregation server. 
We also tested performance with a gradient vector size of 16,000 to simulate high-dimensional data. 

\item \textit{Communication Bandwidth.} We recorded the bandwidth, i.e., total message size exchanged between clients, assisting nodes, and the server, to evaluate communication costs.
The reason why we use message size instead of transmission time is because message transmission time highly depends on current network conditions. Variations in hardware and network quality can greatly affect latency.
\end{itemize}



\noindent
\textbf{Metric for Performance:}
We assessed the performance of our DP methods by evaluating the \textit{noise multiplier} and the \textit{accuracy} of the final models. Client data allocation adheres to standard protocols~\cite{yang2023privatefl,flamingo}. 
Unless specified otherwise, training proceeds with 5 clients, 50 communication rounds, and a client-sampling fraction of 1.0—in other words, each client contributes to every round. Pilot runs showed that fifty rounds are already enough for the models to converge under DP noise; extending training beyond that point improves accuracy only marginally while increasing computation. 
Using a sampling fraction of one eliminates variability due to partial participation, allowing us to isolate the sole effect of injected noise.
Starting from this baseline, we varied the privacy budget $\epsilon$ values across $\{5, 10, 15, 20\}$ to evaluate the privacy–utility trade-off. We then fixed $\epsilon$ at 10 and varied the two factors independently. 
First, we shortened training to as few as a single iteration and gradually extended it through $\{1, 5, 10, 15, 20\}$ to reveal how DP noise hinders convergence when updates are scarce. Second, with 10 clients in total, we adjusted sampling fraction over $\{0.3, 0.5, 0.7, 0.9\}$ so that different number of clients participate in each iteration, to evaluate the impact of the size of the actively participating subset on both noise levels and final accuracy. 



\section{Evaluation Results}

In this section, we present the experimental results evaluating the efficiency and performance of \sys. Our findings address the following key research questions:





\begin{itemize}

\item \textbf{Section~\ref{sec:eval:eff} (Efficiency Measurement)}: How efficient is post-quantum enhanced secure aggregation in FL? What gains in efficiency does \sys achieve through the use of pre-computation?
\item \textbf{Section~\ref{sec:eval:perf} (Model Performance Measurement)}: How does \sys’s differential privacy module perform under different threat models? What is the impact of DP-enhanced methods on FL performance across different training datasets?
\end{itemize}

\subsection{Efficiency Measurement}~\label{sec:eval:eff}


\subsubsection{Computational Cost}
We systematically compare the computational cost of \sys in the Setup phase and Aggregation phase with four state-of-the-art approaches, i.e., \flamingo, \pqsa, \eseafl, and \micro.


\begin{figure*}[htbp]
    \centering
    \begin{minipage}{0.32\textwidth}
        \centering
        \includegraphics[width=\textwidth]{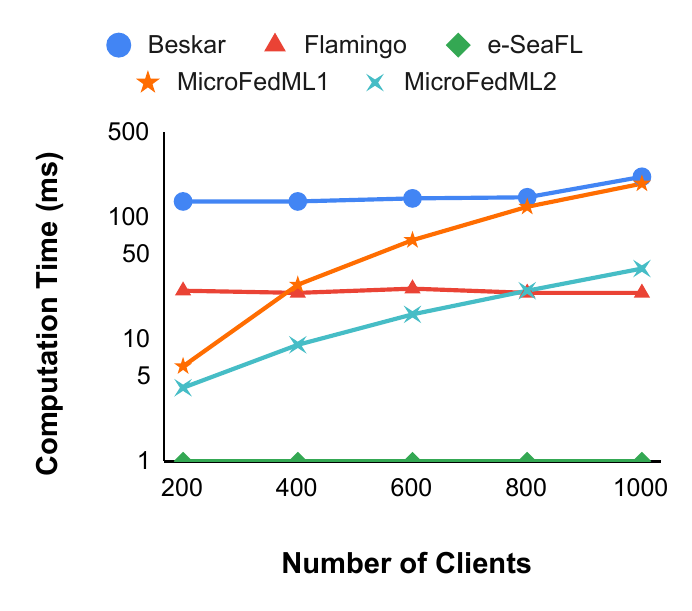}
        \caption*{(a) Server}
    \end{minipage}
    \hfill
    \begin{minipage}{0.32\textwidth}
        \centering
        \includegraphics[width=\textwidth]{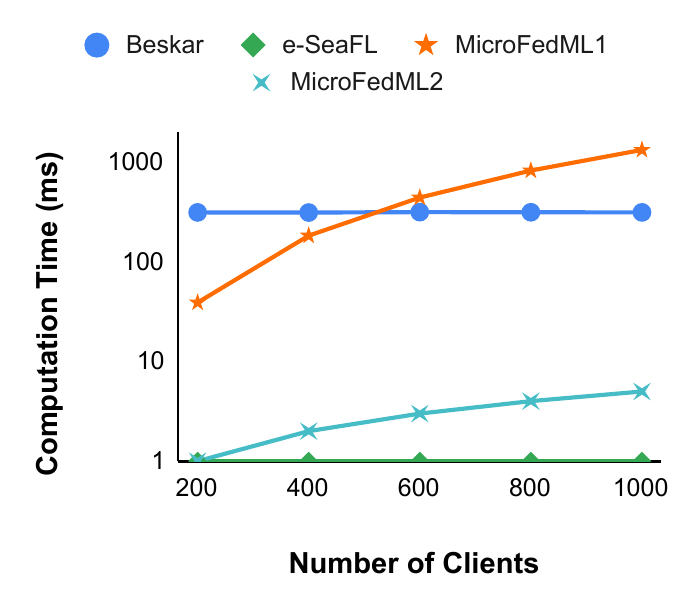}
        \caption*{(b) Client}
    \end{minipage}
    \hfill
    \begin{minipage}{0.32\textwidth}
        \centering
        \includegraphics[width=\textwidth]{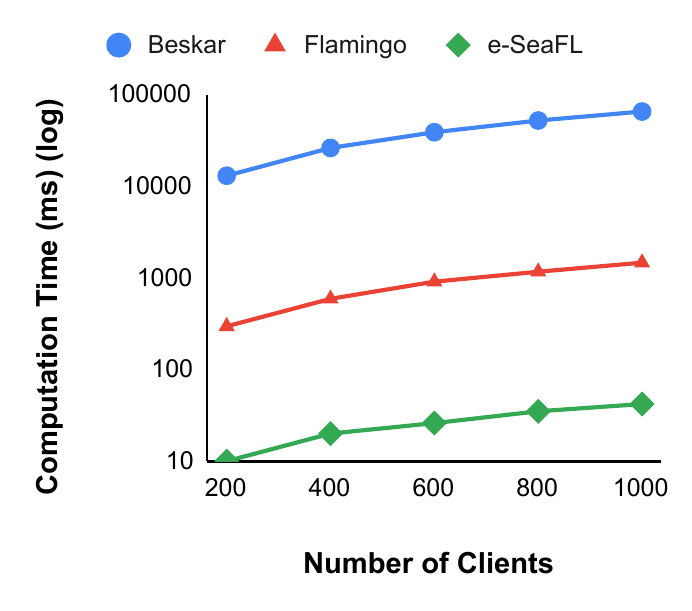}
        \caption*{(c) Assisting node}
    \end{minipage}
    \caption{Computation Time of the Setup phase (The dimension of the weight list is set to 16K.)}
    \label{fig:sa_setup_time}
\end{figure*}

\vspace{5pt}
\noindent
\textbf{Setup Phase.}
Fig.~\ref{fig:sa_setup_time} shows the computation costs of each method during the Setup phase.
\sys stands out for its nearly constant computation time on 
the server (Fig.~\ref{fig:sa_setup_time} (a)) and client (Fig~\ref{fig:sa_setup_time} (b)) side,
even as the number of clients increases.
This is because the operations \sys performs on the client side are primarily dependent on the number of assisting nodes, which typically remains constant in general FL settings, and the fact that server-side operations have a constant computational cost.
In contrast, the computation time for assisting nodes increases as the number of clients grows (Fig.~\ref{fig:sa_setup_time} (c)).
This increase results from tasks performed by assisting nodes, such as shared secret generation, which become more intensive as  the number of clients rises.

Our results also show that, in the setup phase, \sys incurs higher computation times compared to \flamingo, \eseafl and \micro across all entities.
The reason is that \sys incorporates an additional precomputation step in the setup phase, where several digital signature parameters are precomputed for subsequent runtime signing (see Section~\ref{sec:pqfl:sign}).
\pqsa is excluded from this comparison as it only distributes several public parameters without performing significant computations. It is also worth noting that \flamingo does not involve any client-side operations, as it relies on a trusted third party to generate and distribute signing keys for each entity.

\begin{figure*}[htbp]
    \centering
    \begin{minipage}{0.32\textwidth}
        \centering
        \includegraphics[width=\textwidth]{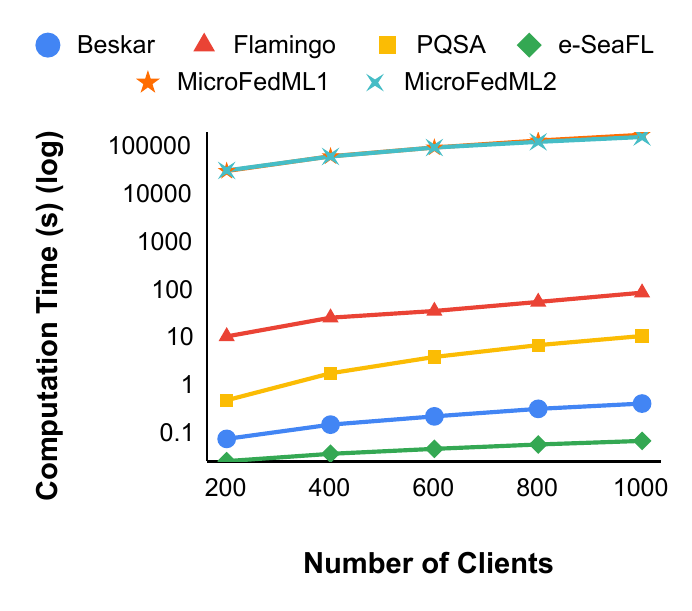}
        \caption*{(a) Server}
    \end{minipage}
    \hfill
    \begin{minipage}{0.32\textwidth}
        \centering
        \includegraphics[width=\textwidth]{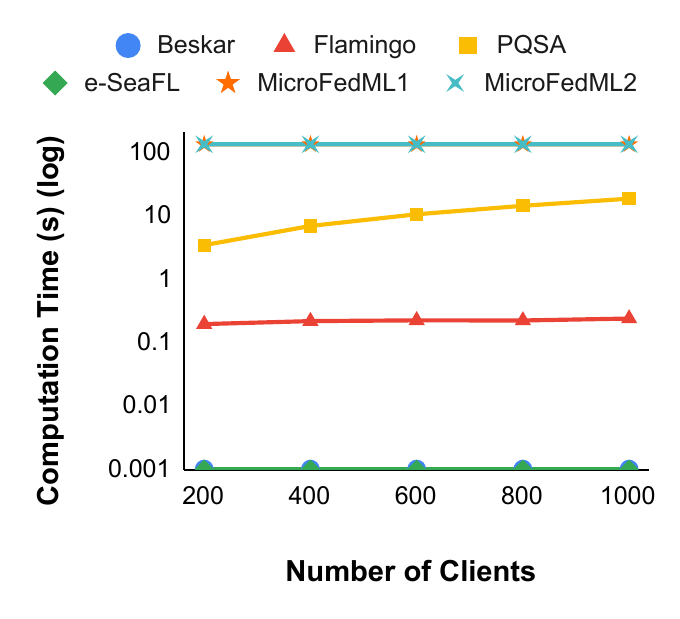}
        \caption*{(b) Client}
    \end{minipage}
    \hfill
    \begin{minipage}{0.32\textwidth}
        \centering
        \includegraphics[width=\textwidth]{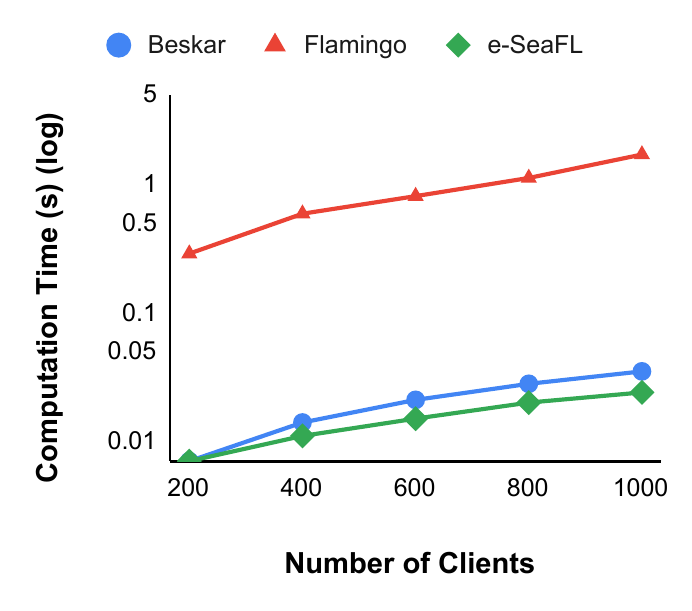}
        \caption*{(c) Assisting node}
    \end{minipage}
    \caption{Computation Time of the Aggregation phase (The dimension of the weight list is set to 16K.)}
    \label{fig:sa_agg_time}
\end{figure*}

\vspace{5pt}
\noindent
\textbf{Aggregation Phase.}
Fig.~\ref{fig:sa_agg_time} shows the computation time of \sys in the aggregation phase compared with that of \flamingo, \pqsa, \eseafl, and \micro.
A key advantage of \sys is that client-side computation time remains constant
(Fig.~\ref{fig:sa_agg_time} (b)), 
, even as the number of clients grows.
This stability minimizes client-side resource demands, making \sys particularly suitable for settings with resource-constrained FL clients.
Moreover, the computation time increases for the server (Fig.~\ref{fig:sa_agg_time} (a)) and assisting nodes (Fig.~\ref{fig:sa_agg_time} (c)) as the number of clients grows.
This increase is primarily due to the need for assisting nodes and the server to verify the signatures of messages sent by clients, making computation time proportional to the number of clients.

The results demonstrate that all entities in \sys incur minimal computational overhead, with computation times significantly lower than those of \flamingo, \pqsa, and \micro, with the exception of \eseafl.
Specifically, \sys has a similar computation time to \eseafl on the client side, but a higher computation time on the assisting nodes and server.
This difference arises because both \sys and \eseafl clients perform similar operations, including gradient masking and a few signing operations. However, the assisting nodes and server in \sys must conduct a large number of signature verifications, which scale with the number of clients. The post-quantum signing algorithm (i.e., Dilithium) used by \sys is originally slower than the traditional algorithm (i.e., ECDSA) used in \eseafl, leading to higher computation times as the number of clients increases.

\sys outperforms \flamingo, \pqsa, and \micro in terms of computation time across clients, assisting nodes, and the server. 
\flamingo requires additional steps to remove masks for dropout clients, necessitating extra information (e.g., shares of secrets used to generate the masks) from clients and decryptors to allow the server to recover the gradient vectors of those offline clients.
\pqsa uses SHPRG-based encryption for gradient masking, which involves an extra step between clients and the server to decrypt and recover the masks and calculate the aggregated result.
In \micro, each client must additionally compute $H(k)^{X_i}$ and $H(k)^{R_i}$, and the server has to recover the gradient vector by calculating the discrete logarithm of $H(k)^{X_i} - H(k)^{R_i}$. Such discrete logarithm calculation is computationally very intensive, especially as the vector dimensions increase.

Overall, compared with the existing approaches, \sys employs a straightforward strategy for dropout clients by ignoring the masked vectors of dropped clients and utilizes a post-quantum algorithm with precomputation for optimization. 
This approach reduces the number of messages, eliminates unnecessary calculations, and significantly reduces time overhead while maintaining quantum security.

\subsubsection{Communication Cost}
We analyze the communication cost, specifically the outbound bandwidth, of \sys during the setup and aggregation phases in comparison to existing approaches.

\begin{figure*}[htbp]
    \centering
    \begin{minipage}{0.32\textwidth}
        \centering
        \includegraphics[width=\textwidth]{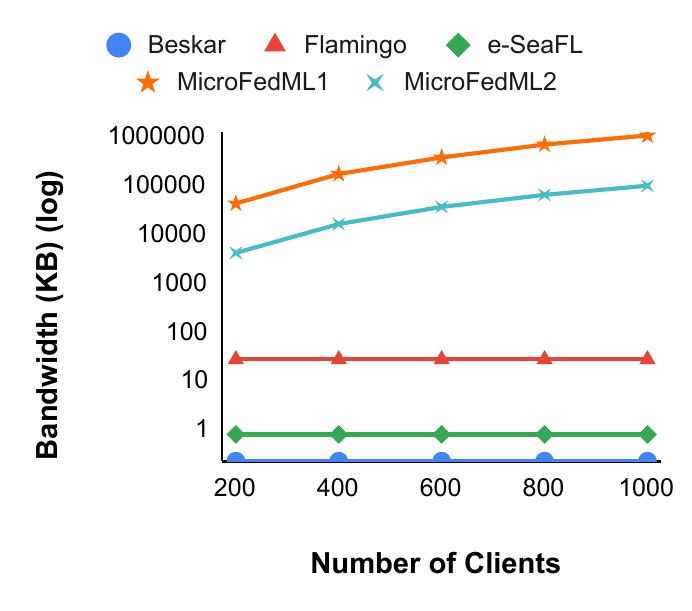}
        \caption*{(a) Server}
    \end{minipage}
    \hfill
    \begin{minipage}{0.32\textwidth}
        \centering
        \includegraphics[width=\textwidth]{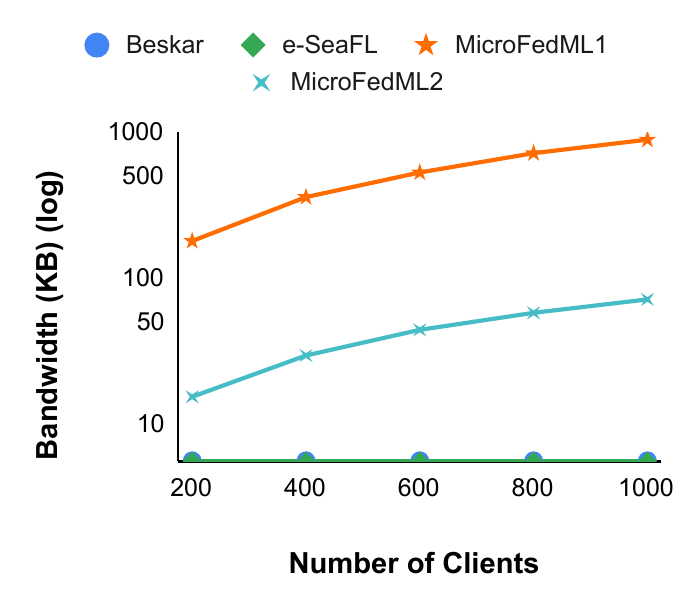}
        \caption*{(b) Client}
    \end{minipage}
    \hfill
    \begin{minipage}{0.32\textwidth}
        \centering
        \includegraphics[width=\textwidth]{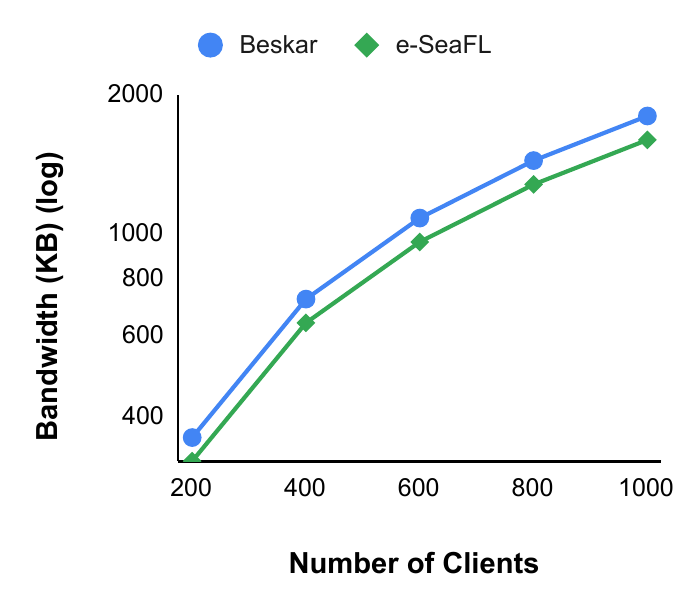}
        \caption*{(c) Assisting node}
    \end{minipage}
    \caption{Bandwidth in the Setup Phase (The dimension of the weight list is set to 16K.)}
    \label{fig:sa_setup_bw}
\end{figure*}

\vspace{5pt}
\noindent
\textbf{Setup Phase.}
Fig.~\ref{fig:sa_setup_bw} presents the outbound bandwidth costs of all entities in \sys during the setup phase. 
In general, the outbound bandwidth of \sys remains constant on the server (Fig.~\ref{fig:sa_setup_bw} (a)) and client (Fig.~\ref{fig:sa_setup_bw} (b)) sides while increases on the assisting nodes (Fig.~\ref{fig:sa_setup_bw} (c)) as the number of clients grows.
This is due to the fact that during the setup phase, the server only needs to broadcast its public key for signing, and each client sends the public keys for key exchange and signing to the assisting nodes, whose number keeps constant in our experimental setting.
However, each assisting node must send these public keys to all clients, making its outbound bandwidth dependent on the number of clients.

\sys incurs less outbound bandwidth than \flamingo and \micro but it is generally comparable to \eseafl across all three entities.
Specifically, during the setup phase, \sys and \eseafl perform the same public key exchange operations, resulting in similar outbound bandwidths, both of which depend on the size of the public keys.
In contrast, in \flamingo, the server and decryptors participate in DKG protocols, where their outbound bandwidth depends on the number and size of the messages in DKG protocols and is proportional to the number of decryptors, but constant relative to the number of clients. This results in greater outbound bandwidth for the server and decryptors in \flamingo compared to \sys, which only involves public key exchanges.
As for \micro, the server has to forward all the communications, i.e., the public keys and secret shares generated by Shamir’s secret sharing scheme, from one client to each of the others, making its bandwidth proportional to the number of clients and higher than that of \sys. Each client in \micro should send its public key and $n-1$ mask shares to all other clients, making its bandwidth increase proportional to the number of clients, which is significantly higher than that of \sys.

Additionally, \pqsa and the clients of \flamingo are excluded from this comparison for the same reasons previously mentioned.

\begin{figure*}[htbp]
    \centering
    \begin{minipage}{0.32\textwidth}
        \centering
        \includegraphics[width=\textwidth]{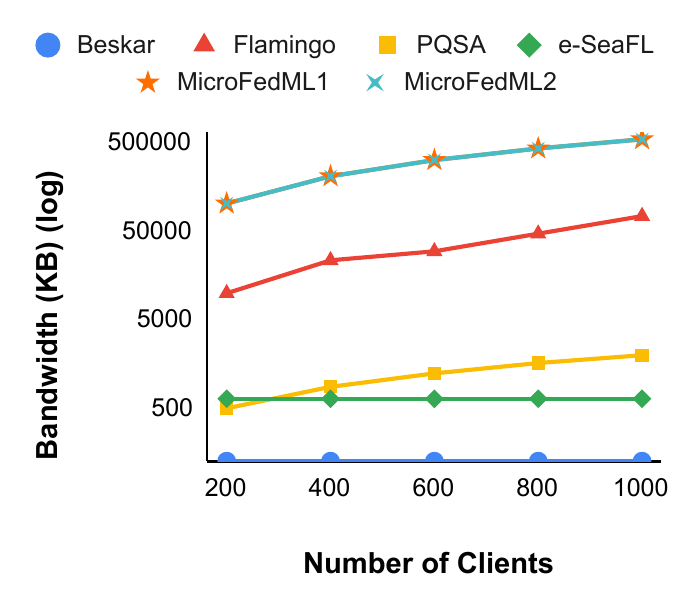}
        \caption*{(a) Server}
    \end{minipage}
    \hfill
    \begin{minipage}{0.32\textwidth}
        \centering
        \includegraphics[width=\textwidth]{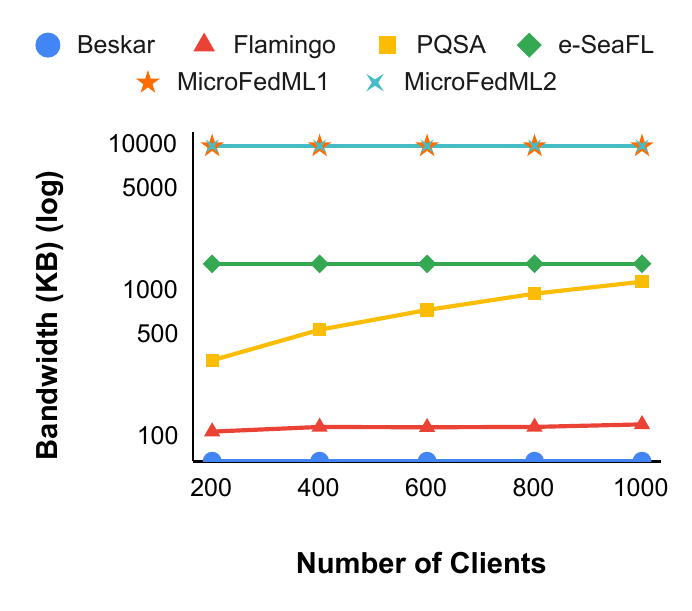}
        \caption*{(b) Client}
    \end{minipage}
    \hfill
    \begin{minipage}{0.32\textwidth}
        \centering
        \includegraphics[width=\textwidth]{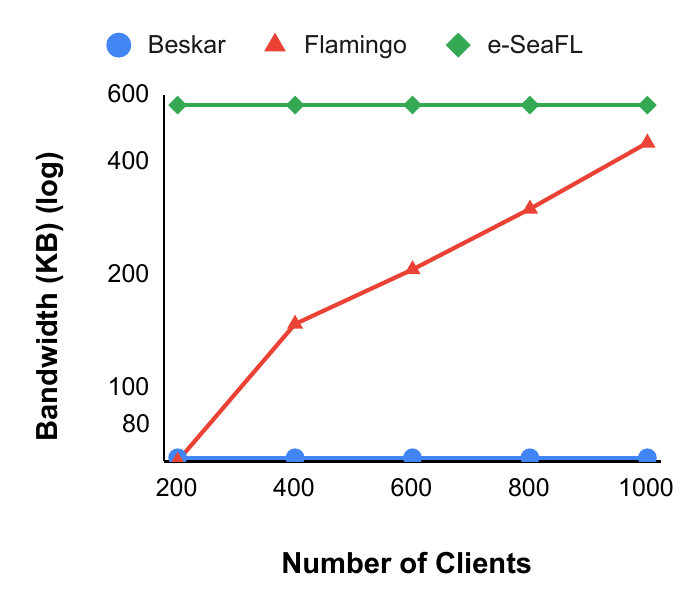}
        \caption*{(c) Assisting node}
    \end{minipage}
    \caption{Bandwidth in the Aggregation phase (The dimension of the weight list is set to 16K.)}
    \label{fig:sa_agg_bw}
\end{figure*}

\noindent
\textbf{Aggregation Phase.}
We present the outbound bandwidth cost of all entities in \sys during the runtime aggregation phase in Fig.~\ref{fig:sa_agg_bw}. 
Generally, 
the outbound bandwidth of \sys remains constant across all three entities, 
i.e., the server (Fig.~\ref{fig:sa_agg_bw} (a)), the clients (Fig.~\ref{fig:sa_agg_bw} (b)), and the assisting nodes (Fig.~\ref{fig:sa_agg_bw} (c)) as the number of clients increases.
This stability occurs because the outbound messages for each entity do not depend on the number of clients. Specifically, each client's outbound messages depend only on the number of assisting nodes and the server, which remain constant in our experimental setup. Assisting nodes send a single message (containing the sum of partial masks) to the server, and the server broadcasts the aggregated results, keeping their outbound bandwidth unchanged.

When compared with existing approaches, \sys demonstrates lower outbound bandwidth costs across all entities compared to \flamingo, \pqsa, \eseafl, and \micro.
Specifically, in \flamingo, the server is involved in forwarding the messages between the clients and decryptors, leading to a higher outbound bandwidth that is roughly proportional to the number of clients.
Clients have to send secret shares to decryptors and neighbors (i.e., clients in the same group), in addition to masked updates to the server, which increases their outbound bandwidth proportional to the number of decryptors and clients (i.e., neighbors), which is significantly higher than \sys.
Decryptors also transmit decrypted secret shares of each client to the server, resulting in higher bandwidth proportional to the number of clients, compared to \sys.
For \pqsa, the server performs an additional communication step with clients to decrypt and recover the masks, resulting in a higher bandwidth that is proportional to the number of clients.
Similarly, clients have to additionally send the decrypted secrets, except for the masked updates, leading to a higher bandwidth, which is proportional to the number of clients.
While \eseafl performs the similar operations to \sys across all entities, the difference in outbound bandwidth is primarily due to the size of each message, especially the size of signatures.
As for \micro, similar to \flamingo, the server has to forward the signature of each client to other clients, incurring a higher outbound bandwidth than \sys, which is proportional to the number of clients. 
Moreover, clients in \micro send masks of online clients along with the masked updates to the server, resulting in bandwidth higher than in \sys.

\subsubsection{Precomputation Improvement}

\renewcommand{\arraystretch}{1.2}
\begin{table}[]
\centering
\caption{Results of Precomputation Improvement (Computation Time in the Aggregation Phase) (ms): P - Precomputation}
\label{table:precomputed}
\begin{tabular}{|c|c|c||c|c||c|c|}
\hline
\multirow{2}{*}{\textbf{N}} & \multicolumn{2}{c||}{\textbf{Client}} & \multicolumn{2}{c||}{\textbf{Server}} & \multicolumn{2}{c|}{\textbf{Assisting Node}} \\ \cline{2-7}
& w/o P & w/ P & w/o P & w/ P & w/o P & w/ P \\
\hline

200  & 135 & 1 & 77  & 76  & 8767 & 7  \\ \hline
400  & 134 & 1 & 163 & 151 & 17703 & 14  \\ \hline
600  & 135 & 1 & 227 & 225 & 27008 & 21  \\ \hline
800  & 134 & 1 & 323 & 302 & 35281 & 28  \\ \hline
1000 & 134 & 2 & 417 & 377 & 43144 & 35 \\ \hline

\end{tabular}
\end{table}

To assess the impact of precomputation strategies on \sys, we conducted a comparative experiment implementing \sys without precomputation in the Dilithium signing algorithms and mask generation. The results, shown in Table~\ref{table:precomputed}, indicate that precomputation significantly enhances efficiency for both assisting nodes and clients as the number of clients increases.

Without precomputation, an assisting node must perform additional mask generation operations for $n$ clients and one signing operation, while a client must perform additional mask generation for $k$ assisting nodes and $k+1$ signing operations. Consequently, precomputation notably reduces the computational burden on these entities.

On the server side, the benefits of precomputation are less pronounced. The reason is that the server performs relatively fewer computational operations, limited to only one signing operation, compared to the assisting nodes and clients.

\subsection{Feasibility of Employing Assisting Nodes}
Following our computational and communication evaluation of \sys~and its counterpart, we emphasize that a assisting node incurs only minimal additional overhead—specifically, about \uline{6 ms more computation time} than a regular client (with a total of 200 clinets), and \uline{approximately 2 KB less in communication}. Given this lightweight cost, a simple distributed ledger-based method can be used to select a set of clients in each iteration to serve as assisting nodes.
This is not a strong assumption since sharing metadata about the model and training parameters is common in FL. Deploying assisting nodes significantly reduces the trust assumption because the privacy of the protocol holds as long as at least one assisting node is honest. At the same time, it minimizes the overhead on participating users and encourages broader participation in the learning process.

\subsection{Model Performance}~\label{sec:eval:perf}

\renewcommand{\arraystretch}{1.2}
\begin{table*}[t]
\centering
\caption{Results of Model Performance under Different Differential Privacy Parameters.}
\label{tab:dp_acc}
\begin{tabular}{|c|c|c||c|c|c||c|c|c||c|c|c||c|c|c||c|c|c|}
\hline
\multirow{2}{*}{\textbf{$\epsilon$}} & 
\multirow{2}{*}{\textbf{\#Iters}} &
\multirow{2}{*}{\textbf{\#Cli}} &
\multicolumn{3}{c||}{\textbf{MNIST}} & 
\multicolumn{3}{c||}{\textbf{EMNIST}} & 
\multicolumn{3}{c||}{\textbf{CIFAR-10}} & 
\multicolumn{3}{c||}{\textbf{CIFAR-100}} &
\multicolumn{3}{c|}{\textbf{CHMNIST}} \\ \cline{4-18}
& & 
& \textbf{NDP} & \textbf{CDP} & \textbf{LDP} 
& \textbf{NDP} & \textbf{CDP} & \textbf{LDP}
& \textbf{NDP} & \textbf{CDP} & \textbf{LDP}
& \textbf{NDP} & \textbf{CDP} & \textbf{LDP}
& \textbf{NDP} & \textbf{CDP} & \textbf{LDP} \\ \hline

20 & 50 & 5/5 & 0.99 & 0.98 & 0.86 & 0.99 & 0.98 & 0.87 & 0.83 & 0.81 & 0.45 & 0.76 & 0.63 & 0.02 & 0.84 & 0.83 & 0.27 \\ \hline
15 & 50 & 5/5 & 0.99 & 0.97 & 0.84 & 0.99 & 0.98 & 0.85 & 0.83 & 0.80 & 0.23 & 0.76 & 0.56 & 0.02 & 0.84 & 0.81 & 0.26 \\ \hline
10 & 50 & 5/5 & 0.99 & 0.87 & 0.78 & 0.99 & 0.90 & 0.82 & 0.83 & 0.77 & 0.20 & 0.76 & 0.23 & 0.02 & 0.84 & 0.76 & 0.23 \\ \hline
5  & 50 & 5/5 & 0.99 & 0.77 & 0.71 & 0.99 & 0.80 & 0.77 & 0.83 & 0.57 & 0.17 & 0.76 & 0.02 & 0.01 & 0.84 & 0.52 & 0.22 \\ \hline \hline

10 & 1  & 5/5 & 0.95 & 0.09 & 0.09 & 0.94 & 0.10 & 0.09 & 0.38 & 0.28 & 0.10 & 0.02 & 0.02 & 0.01 & 0.40 & 0.42 & 0.11 \\ \hline
10 & 5  & 5/5 & 0.96 & 0.54 & 0.53 & 0.94 & 0.41 & 0.71 & 0.68 & 0.60 & 0.20 & 0.64 & 0.18 & 0.01 & 0.64 & 0.65 & 0.13 \\ \hline
10 & 10 & 5/5 & 0.98 & 0.91 & 0.72 & 0.95 & 0.91 & 0.77 & 0.78 & 0.68 & 0.20 & 0.73 & 0.25 & 0.02 & 0.73 & 0.68 & 0.23 \\ \hline
10 & 15 & 5/5 & 0.98 & 0.92 & 0.77 & 0.97 & 0.93 & 0.80 & 0.78 & 0.72 & 0.10 & 0.75 & 0.30 & 0.02 & 0.79 & 0.71 & 0.26 \\ \hline
10 & 20 & 5/5 & 0.98 & 0.93 & 0.79 & 0.98 & 0.94 & 0.83 & 0.79 & 0.73 & 0.24 & 0.76 & 0.33 & 0.02 & 0.81 & 0.72 & 0.27 \\ \hline \hline

10 & 50 & 3/10 & 0.99 & 0.48 & 0.78 & 0.99 & 0.56 & 0.82 & 0.83 & 0.73 & 0.23 & 0.72 & 0.43 & 0.02 & 0.84 & 0.72 & 0.22 \\ \hline
10 & 50 & 5/10 & 0.99 & 0.87 & 0.78 & 0.99 & 0.90 & 0.82 & 0.83 & 0.77 & 0.20 & 0.74 & 0.44 & 0.02 & 0.84 & 0.76 & 0.23 \\ \hline
10 & 50 & 7/10 & 0.99 & 0.95 & 0.77 & 0.99 & 0.96 & 0.81 & 0.81 & 0.79 & 0.15 & 0.75 & 0.45 & 0.02 & 0.84 & 0.81 & 0.27 \\ \hline
10 & 50 & 9/10 & 0.99 & 0.97 & 0.75 & 0.99 & 0.97 & 0.80 & 0.80 & 0.79 & 0.13 & 0.75 & 0.45 & 0.01 & 0.84 & 0.81 & 0.27 \\ \hline 

\end{tabular}
\end{table*}

\begin{figure}[t]
    \centering
    \includegraphics[width=0.5\textwidth]{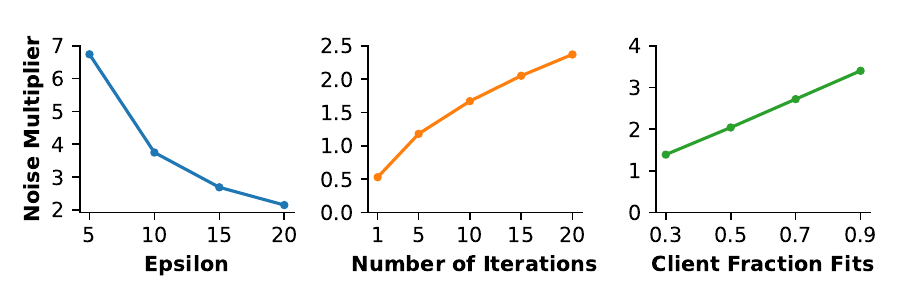}
    \caption{Effect of Epsilon, Iterations, and Clients on Noise.}
    \label{fig:noise}
\end{figure}


Table~\ref{tab:dp_acc} presents accuracy results on MNIST, EMNIST, CIFAR-10, CIFAR-100, and CHMNIST under Non-Differential Privacy (NDP), Central Differential Privacy (CDP), and Local Differential Privacy (LDP) across varying privacy budgets $\epsilon$, training iterations, and the fraction of participating clients (fraction fit). Fig.~\ref{fig:noise} shows the corresponding noise multipliers under these configurations. 
Below, we discuss how $\epsilon$, the number of training iterations, and the fraction fit affect both model accuracy and noise. 

\subsubsection{Effect of Privacy Budget ($\epsilon$)}

We first examine how varying privacy budget $\epsilon$ impacts both the noise multiplier and the model accuracy. As expected, NDP, which injects no noise, yields the highest accuracy across all datasets. Once privacy is enforced, both CDP and LDP see reduced accuracy relative to NDP, with CDP generally outperforming LDP.

\noindent
\textbf{Noise Multiplier vs. $\epsilon$.}
Fig.~\ref{fig:noise} underscores the inverse relationship between $\epsilon$ and the noise multiplier. Smaller $\epsilon$ values offer stronger privacy guarantees but require to increase the noise multiplier, which injects more noise at each training iteration. This additional noise hampers the model’s ability to learn, explaining the sharper accuracy decline observed for lower $\epsilon$.

\noindent
\textbf{Model Accuracy vs. $\epsilon$.} 
Table~\ref{tab:dp_acc} shows that accuracy under both CDP and LDP improves as $\epsilon$ increases. At a relatively high privacy budget (e.g., $\epsilon$), the noise multiplier is moderate, and the MNIST model under CDP achieves an accuracy of 0.98—close to the NDP baseline of 0.99—while LDP reaches 0.86. However, tightening privacy (e.g., $\epsilon=5$) forces a much larger noise multiplier, causing significant accuracy drops and widening the gap between the DP methods and NDP.

\subsubsection{Sensitivity Analysis}

We next explore how other key hyperparameters—the number of training iterations and the fraction fit of clients—further influence the trade-off between model accuracy and noise levels. For these experiments, we fix $\epsilon$ while varying one parameter at a time.

\noindent
\textbf{Number of Training Iterations.}
Increasing the number of training iterations substantially improves model convergence, even under noisy conditions. When training for only one iteration under LDP on MNIST with $\epsilon=10$, accuracy can be just under $10\%$. In contrast, extending the training to 50 iterations under the same $\epsilon$ raises accuracy to above $70\%$. This indicates that additional iterations provide the model more opportunity to learn from the data, thereby offsetting much of the performance degradation caused by noise. 
Although longer training may require additional noise injections (i.e., larger noise multiplier),
the results show that the convergence gains generally outweigh the potential downside of increased noise, particularly when training has not yet converged.

\noindent
\textbf{Fraction fit of Clients.}
The fraction of clients selected to participate in each training iteration—rather than merely the total pool of clients—strongly influences both the noise multiplier and model accuracy. Under CDP, noise is injected centrally after aggregating client updates. When a larger fraction of the total clients participate, the training process benefits from more data per iteration, typically boosting accuracy. However, this increase in participating clients can also raise the noise multiplier, since a larger collective update requires proportionally more noise to preserve the privacy budget.  
Table~\ref{tab:dp_acc} shows that as this fraction increases, more data is included in each iteration, which improves convergence and often offsets any added noise, ultimately leading to higher accuracy for CDP.
By contrast, in LDP, each client adds noise directly to its own data or gradients. Hence, a larger fraction of participating clients produces a greater cumulative amount of noise in the aggregated update, which can reduce accuracy gains from using more data. 
As a result, we observed in Table~\ref{tab:dp_acc}, for LDP, larger client fractions can lead to lower accuracy.


\vspace{5pt}
In summary, both central differential privacy (CDP) and local differential privacy (LDP) incur some loss of accuracy relative to the non-private baseline (NDP); however, for tasks of moderate complexity this reduction remains modest, allowing both schemes to provide competitive performance while markedly strengthening data protection. Therefore, one can steer the privacy–utility trade-off by adjusting the privacy budget $\epsilon$, the number of training rounds, and the fraction of participating clients, thereby meeting application-specific privacy requirements without incurring unnecessary accuracy loss.

On more demanding benchmarks such as CIFAR-100, the strong noise injected under LDP hampers fine-grained feature learning and produces a noticeably sharper accuracy decline. 
Although recent study~\cite{yang2023privatefl} reports substantially higher CIFAR-100 accuracy with LDP, we took a closer examination which reveals that the reported gains hinge on non-standard model architectures, and the publicly released code, when adapted to our standardized experimental setup, fails to reproduce the claimed performance.

\section{{Related Work}}

Research on FL's security and privacy focuses on two main streams: addressing attacks targeting the privacy of the training data (\cite{zheng2022aggregation,bonawitz2017practical,bell2020secure}) and preventing client-side attacks undermining the model's reliability (\cite{bell2023acorn,roy2022eiffel}). 
Privacy-preserving methods can be categorized further into those that protect privacy during training and those that protect privacy post-training. 
Techniques safeguarding privacy during training focus on protecting user updates during the training phase~\cite{wang2024model}. They often leverage multi-party computation (MPC) or homomorphic encryption to allow the central server to compute the model in each iteration without accessing individual user updates. 
However, the direct application of these methods often fails to efficiently account for user dropouts, an inherent challenge in FL, where mobile participants unexpectedly exit the training process due to factors such as low battery or unstable network connections. 
A series of follow-up works have focused on designing more resilient, secure aggregation methods capable of handling user dropouts~\cite{zheng2022aggregation,eseafl,flamingo}.

Post-training privacy mechanisms rely on DP to address the deployed model's privacy leakage. These methods often focus on computing a carefully calibrated noise to be injected either at the user updates or the intermediate model to prevent the leakage of sensitive information from the deployed model. 

The other stream of research focuses on ensuring the reliability of the final model. Since the global model is derived from user updates, an adversary can compromise its reliability by injecting malicious updates. To mitigate these attacks, various input validation mechanisms have been proposed with the aim of ensuring that only legitimate updates are used to compute the final model. This can be done by leveraging various methods such as enforcing norm bounds~\cite{lycklama2023rofl}, the use of trusted hardware~\cite{zhao2021sear, guo2020v}, etc.

While existing approaches typically treat secure aggregation and DP as separate solutions, Beskar is the first to integrate both within a unified framework guided by a comprehensive threat model for FL. To the best of our knowledge, it is also the first protocol to address the emerging challenge of post-quantum security while preserving practical efficiency and model performance, particularly in resource-constrained mobile environments.

\section{Conclusion}
In this paper, we propose \sys, a secure framework to enhance privacy protection in FL by integrating post-quantum secure aggregation with differential privacy. We introduce novel precomputation techniques to optimize the efficiency of post-quantum signing algorithms and mask calculation for secure aggregation. Additionally, we present a comprehensive threat model for FL, addressing various adversarial scenarios and systematically applying differential privacy strategies to offer tailored defenses against diverse privacy threats. Through detailed efficiency measurements and performance analysis, we demonstrate that \sys achieves a balanced trade-off between security, computational efficiency, and model accuracy. Our contributions represent a significant advancement in making FL both secure and practical in a post-quantum world, addressing contemporary privacy challenges while adhering to regulatory requirements.

\section*{Acknowledgment}
The authors would like to express their sincere gratitude to the editor and reviewers for their guidance, insightful comments, and constructive feedback, which have significantly improved the quality of this paper.  s
This work reported in this paper has been supported by the National Science Foundation (NSF) under Grant No. 2112471, and by the NSF–SNSF joint program under Grant No. ECCS 2444615.


\ifCLASSOPTIONcaptionsoff
  \newpage
\fi



%


\bibliographystyle{IEEEtran}
\bibliography{bibtext,PPML}

%







\appendices
\newpage
\pagebreak
\section{Detailed Signing Algorithms}~\label{sec:app:algo}

\begin{algorithm}[]
\centering

\raggedright{\textbf{$(\sk_\kem,\pk_\kem)\gets \sgnkeygen(1^\kappa)$}}
\vspace{-2mm}

\hrulefill 
\begin{algorithmic}[1]
\State $\mathbf{A} \leftarrow \mathbb{R}_q^{k \times \ell}$ 
\State $(\mathbf{s}_1, \mathbf{s}_2) \leftarrow \mathbb{S}_\eta^{\ell} \times \mathbb{S}_\eta^{k}$ 
\State $\mathbf{t} := \mathbf{A}\mathbf{s}_1 + \mathbf{s}_2$ 
\State \textbf{return} $(\textit{pk} = (\mathbf{A}, \mathbf{t}), \textit{sk} = (\mathbf{A}, \mathbf{t}, \mathbf{s}_1, \mathbf{s}_2))$ 
\end{algorithmic}

\hrulefill \vspace{-1mm}

\textbf{$\sigma \gets \sgnsign(\sk_\Pi, m)$}
\vspace{-2mm}

\hrulefill

\begin{algorithmic}[1]


\State $\mathbf{z} := \perp$ 
\State \textbf{while} $\mathbf{z} = \perp$ \textbf{do} 
\State \hspace{\algorithmicindent} $\mathbf{y} \leftarrow \mathbb{S}_{\gamma_1 - 1}^{\ell}$ 
\State \hspace{\algorithmicindent} $\mathbf{w}_1 := \mathsf{HighBits}(\mathbf{A}\mathbf{y}, 2\gamma_2)$ 
\State \hspace{\algorithmicindent} $c \in \mathbb{B}_{60} := \mathcal{H}(m \parallel \mathbf{w}_1)$ 
\State \hspace{\algorithmicindent} $\mathbf{z} := \mathbf{y} + c\mathbf{s}_1$ 
\State \hspace{\algorithmicindent} \textbf{if} $\|\mathbf{z}\|_\infty \geq \gamma_1 - \beta$ 
\State \hspace{\algorithmicindent} \hspace{1em} \textbf{or} $\|\mathsf{LowBits}(\mathbf{A}\mathbf{y} - c\mathbf{s}_2, 2\gamma_2)\|_\infty \geq \gamma_2 - \beta$, 
\State \hspace{\algorithmicindent} \hspace{1em} \textbf{then} $\mathbf{z} := \perp$ 
\State  \textbf{return} $\sigma = (\mathbf{z}, c)$ 
\end{algorithmic}

\hrulefill \vspace{-1mm}

\textbf{$\{0,1\}\gets \sgnverify(\pk_\Pi,m,\sigma)$}
\vspace{-2mm}

\hrulefill

\begin{algorithmic}[1]
\State  $\mathbf{w}_1' := \mathsf{HighBits}(\mathbf{A}\mathbf{z} - c\mathbf{t}, 2\gamma_2)$ 
\State  \textbf{if} $\|\mathbf{z}\|_\infty < \gamma_1 - \beta$ \textbf{and} $c = \mathcal{H}(m \parallel \mathbf{w}_1')$ \textbf{then return} 1 
\State  \textbf{else return} 0
\end{algorithmic}
\caption{Dilithium}\label{alg:dilithium}
\end{algorithm}
\begin{algorithm}
\caption{Dilithium with Precomputation}\label{alg:dilifull}

\textbf{Precomputing: $\mathcal{LS} \gets \precomputesgn(\sk_\Pi, N)$}

\raggedright{\textbf{Input:} Dilithium private key $\sk_\Pi$ and the number of groups for Dilithium parameters to be precomputed $N$.}

\textbf{Output:} $\mathcal{LS}$, in which $\mathcal{LS}[\sk_\Pi]$ including $N$ groups of precomputed Dilithium parameters for $\sk_\Pi$

\hrulefill 

\begin{algorithmic}[1]
\State $\mathcal{LS}[\sk_\Pi] := \emptyset$ 
\State $\mathbf{A} \in R_{q}^{k \times l} := \mathsf{ExpandA}(\rho)$
\State $u \in \{0,1\}^{384} := \mathsf{CRH}(tr)$

\Comment{\mal{In original Dilithium: $\mu \gets \mathsf{CRH}(tr||m)$}}
\State $\kappa := 0$
\While{\text{Length of} $\mathcal{LS}[\sk_\Pi] < N$} 
    \State $y \in \mathcal{S}^l_{\gamma_1 - 1} := \mathsf{ExpandMask}(K \parallel u \parallel \kappa)$
    
    \Comment{\mal{Using $u$ not $\mu$}}
    \State $\mathbf{w} := \mathbf{Ay}$
    \State $\mathbf{w_1} := \mathsf{HighBits}(\mathbf{w}, 2\gamma_2)$
    \If{($A, u, y, w, w_1$) not in $\mathcal{LS}[\sk_\Pi]$}
        \State Add ($A, u, y, w, w_1$) into $\mathcal{LS}[\sk_\Pi]$
    \EndIf
    \State $\kappa := \kappa + 1$
\EndWhile
\State \textbf{return} $\mathcal{LS}$ 
\end{algorithmic}

\vspace{-1mm}
\hrulefill

\textbf{Signing with precomputation: $\sigma \gets \presgnsign(\sk_\Pi, m, \mathcal{LS})$} 

\raggedright{\textbf{Input:} Dilithium private key $\sk_\Pi$, message to be signed $m$, and the list of precomputed Dilithium parameters $\mathcal{LS}$.}

\textbf{Output:} the signature $\sigma$.

\hrulefill

\begin{algorithmic}[1]

\For{$(A, u, y, w, w_1) \text{ in } \mathcal{LS}[\sk_\Pi]$}
    \State $c \in B_{60} := H(m \parallel u \parallel w_1)$
    \State $\mathbf{z} := \mathbf{y} + c\mathbf{s_1}$
    \State $(\mathbf{r_0}, \mathbf{r_1}) := \mathsf{Decompose}(\mathbf{w}-c\mathbf{s_2}, 2\gamma_2)$
    \If{$\lVert z\rVert_{\infty} \leq \gamma_1 - \beta \text{ and } \lVert \mathbf{r_0} \rVert_{\infty} \leq \gamma_2 - \beta$}
        \State Remove $(A, u, y, w, w_1)$ from $\mathcal{LS}[\sk_\Pi]$
        \State \Return{$\sigma = (z, c)$}
    \EndIf
\EndFor

\Comment{\mal{If none of the precomputed parameters are satisfied, invoke the original Dilithium signing algorithm.}}
\State $\mu \gets \mathsf{CRH}(tr||m)$
\State $z := \bot$
\While{$z = \bot$}
    \State $y \in \mathcal{S}^l_{\gamma_1 - 1} := \mathsf{ExpandMask}(K \parallel \mu \parallel \kappa)$
    \State $\mathbf{w} := \mathbf{Ay}$
    \State $\mathbf{w_1} := \mathsf{HighBits}(\mathbf{w}, 2\gamma_2)$
    
    \State $c \in \mathcal{B}^{60} := H(\mu \parallel w_1)$
    \State $\mathbf{z} := \mathbf{y} + c\mathbf{s_1}$
    \State $(\mathbf{r_0}, \mathbf{r_1}) := \mathsf{Decompose}(\mathbf{w}-c\mathbf{s_2}, 2\gamma_2)$
    \If{$\lVert z\rVert_{\infty} > \gamma_1 - \beta \text{ or } \lVert \mathbf{r_0} \rVert_{\infty} > \gamma_2 - \beta$}
        \State $z := \bot$
    \EndIf
\EndWhile
\State \Return{$\sigma = (z, c)$}

\end{algorithmic}
\end{algorithm}

The optimized signing algorithms with precomputation (Algorithm~\ref{alg:dilifull}) build upon the Dilithium algorithms~\cite{dilithium} (Algorithm~\ref{alg:dilithium}). 
The precomputation process aims to enhance signing efficiency by precomputing and storing parameters related to the private key $\sk_\Pi$, reducing the need for new randomness and complex lattice computations during each runtime signing operation. 
The procedure iteratively computes values, including a matrix $\mathbf{A}$ (derived from a seed $\rho$) and a masking vector $\mathbf{y}$, generated through the $\mathsf{ExpandMask}$ function. This function combines the private key and random inputs to expand values. The computed sets are stored in $\mathcal{LS}[\sk_\Pi]$ until $N$ precomputed groups are achieved.

During signing, the precomputed values in $\mathcal{LS}[\sk_\Pi]$ are used to generate signatures more efficiently. The algorithm iterates through $\mathcal{LS}[\sk_\Pi]$, checking for compliance with validity conditions defined by threshold parameters $(\gamma_1, \gamma_2)$. When a valid set is identified, it is used for signature construction, and that set is removed from the precomputed list.
If no precomputed sets meet the conditions, the algorithm reverts to the original Dilithium signing method, computing parameters in real-time. This fallback ensures that the signing process remains functional, even when precomputed values are unavailable or unsuitable.

\section{Cryptographic Security Notions}\label{sec:cryptonotions}

\begin{definition}[EU-CMA]\label{def:EUCMA}
Existential Unforgeability under Chosen Message Attack (EU-CMA) experiment $Expt^{\text{EU-CMA}}_{\Pi,A}$  for  a signature scheme $\Pi$:$(\sgnkeygen$, $\sgnsign$, $\sgnverify)$ with an adversary $A$ is defined as follows. 

\begin{itemize}
    \item $(\sk,\pk)\gets\Pi.\sgnkeygen(1^\kappa)$
    \item $(m^*,\sigma^*)\gets A^{\Pi.\sgnsign(\cdot)}(\pk)$
    \item If $1\gets\Pi.\sgnverify(\cdot)$ and $m^*$ was never queries to $\sgnverify(\cdot)$, return `success' otherwise, output `$\bot$'. 
    
\end{itemize}

\end{definition}

\begin{definition}[IND-CCA]\label{def:INDCCA}
The indistinguishability of a key encapsulation scheme  $\kem=\{\sgnkeygen, \kemenc, \kemdec\}$ with key space $\mathcal{K}$ under chosen ciphertext attack (IND-CCA) experiment $Expt^{\text{IND-CCA}}_{\kem,A}$ with an adversary $A$ is defined as follows. 

\begin{itemize}
    \item $(\sk,\pk)\gets\Pi.\sgnkeygen(1^\kappa)$
     \item $b\gets \{0,1\},(c_{x},x_0)\gets\Pi.\kemenc(\pk)$
     \item $x_1\gets\mathcal{K}$
    \item $b'\gets A^{\Pi.\kemenc(\cdot),\Pi.\kemdec(\cdot)}(\pk,c_{x},x_b)$

\end{itemize}
The advantage of the adversary in the above experiment is defined as $\Pr[b=b']\leq \frac{1}{2} +\varepsilon$ for a negligible $\varepsilon$. 
\end{definition}


\end{document}